\documentclass[journal,11pt,onecolumn,draftclsnofoot]{IEEEtran}

\usepackage[T1]{fontenc}
\usepackage[latin1]{inputenc}

\usepackage[final]{microtype}
\usepackage{exscale}
\usepackage{soul}

\usepackage{array}
\usepackage{hhline}
\usepackage{arydshln}

\usepackage{amssymb,amsmath,calc}
\usepackage{amsfonts}
\usepackage{dsfont}
\usepackage{mathrsfs}
\usepackage{units}
\usepackage{colonequals}
\usepackage{mathtools}
\usepackage{epsfig}
\usepackage{tikz}
\usetikzlibrary{arrows}
\usetikzlibrary{intersections}
\usetikzlibrary{patterns}
\usepackage{tabularx}
\usepackage{multirow}
\usepackage{paralist}
\usepgflibrary{shapes.callouts}
\usepackage{float} 

\usepackage[hidelinks,breaklinks]{hyperref}
\usepackage[hyperref, amsmath, thref, thmmarks]{ntheorem}
\usepackage{ifthen}
\usepackage{scrtime}
\usepackage{rotating}
\usepackage{smartref}
\usepackage{cleveref}

\usepackage{twoopt}
\usepackage[mathscr]{euscript}

\usepackage{accents}
\usepackage{setspace}

\usepackage{caption}

\makeatletter

\newtheoremstyle{changeMMdef}%
{\item{\theorem@headerfont ##1\ ##2 \theorem@separator} }%
{\item{\theorem@headerfont ##1\ ##2 \normalfont\bfseries\unboldmath(##3)\theorem@separator}  }%
\newtheoremstyle{changeMMthm}%
{\item{\theorem@headerfont ##1\ ##2 \theorem@separator}  }%
{\item{\theorem@headerfont ##1\ ##2 \normalfont\bfseries\unboldmath\itshape(##3)\theorem@separator} }%
\newtheoremstyle{changeMMnumpar}%
{\item{\theorem@headerfont ##2} }%
{\item{\theorem@headerfont \bfseries ##3\theorem@separator} }
\newtheoremstyle{emptyMM}%
{\item{}}%
{\item{\hskip\labelsep\relax ##3\theorem@separator}}%

\makeatother

\theoremheaderfont{\normalfont\bfseries}
\theoremseparator{.}
\theoremnumbering{arabic}
\theoremsymbol{}

\theoremstyle{changeMMthm}
\theorembodyfont{\itshape}

\newtheorem{theorem}{Theorem}

\newtheorem{lemma}{Lemma}

\newtheorem{corollary}{Corollary}

\newtheorem{proposition}{Proposition}

\theoremstyle{changeMMdef}
\theorembodyfont{\normalfont}

\newtheorem{remark}{Remark}

\theoremstyle{changeMMnumpar}
\theorembodyfont{\normalfont}

\newtheorem{numpar}{}

\theoremstyle{emptyMM}
\theorembodyfont{\normalfont}

\theoremstyle{nonumberplain}
\theoremheaderfont{\itshape}
\theorembodyfont{\normalfont}
\theoremsymbol{\ensuremath{\scriptstyle\Box}}

\newtheorem{proof}{Proof}

\addtoreflist{theorem}

\newcommand{\wrt}{w.\,r.\,t.}
\newcommand{\ie}{i.\,e.}
\newcommand{\eg}{e.\,g.}
\crefname{numpar}{paragraph}{paragraphs}
\Crefname{numpar}{Paragraph}{Paragraphs}

% ...

\newcommand{\N}{\ensuremath{\mathds{N}}}
\newcommand{\No}{\ensuremath{\mathds{N}_0}}

\newcommand{\R}{\ensuremath{\mathds{R}}}

\newcommandtwoopt*{\optScript}[3][][]{%
\ifthenelse{\equal{#1}{}}%
{\ifthenelse{\equal{#2}{}}%
	{\ensuremath{#3}}%
	{\ifthenelse{\equal{#2}{'}}%
		{\ensuremath{#3'}}%
		{\ifthenelse{\equal{#2}{''}}%
		{\ensuremath{#3''}}%
		{\ensuremath{#3^{#2}}}%
		}%
	}%
}%
{\ifthenelse{\equal{#2}{}}%
	{\ensuremath{#3_{#1}}}%
	{\ifthenelse{\equal{#2}{'}}%
		{\ensuremath{#3_{#1}'}}%
		{\ifthenelse{\equal{#2}{''}}%
		{\ensuremath{#3_{#1}''}}%
		{\ensuremath{#3_{#1}^{#2}}}%
		}%
	}%
}%
}

\DeclareMathOperator{\rank}{rank}

\newcommandtwoopt*{\meP}[2][][]{\optScript[{#1}][{#2}]{\mathrm{P}}}

\newcommandtwoopt*{\ccaC}[2][][]{\optScript[{#1}][{#2}]{\varrho}}
\newcommandtwoopt*{\rvIdnsI}[2][][]{\optScript[{#1}][{#2}]{\nu}}

\newcommandtwoopt*{\rvIdnsII}[2][][]{\optScript[{#1}][{#2}]{\tilde{\nu}}}
\newcommandtwoopt*{\rvIdnsIII}[2][][]{\optScript[{#1}][{#2}]{\hat{\nu}}}

\newcommand{\transpose}[1]{\ensuremath{#1^{\mathrm{T}}}}

\newcommandtwoopt*{\rvA}[2][][]{\optScript[{#1}][{#2}]{\alpha}}
\newcommandtwoopt*{\rvAt}[2][][]{\optScript[{#1}][{#2}]{\tilde{\alpha}}}
\newcommandtwoopt*{\rvB}[2][][]{\optScript[{#1}][{#2}]{\beta}}
\newcommandtwoopt*{\rvBt}[2][][]{\optScript[{#1}][{#2}]{\tilde{\beta}}}

\newcommandtwoopt*{\rvC}[2][][]{\optScript[{#1}][{#2}]{\varphi}}
\newcommandtwoopt*{\rvX}[2][][]{\optScript[{#1}][{#2}]{\xi}}
\newcommandtwoopt*{\rvXb}[2][][]{\optScript[{#1}][{#2}]{\bar{\xi}}}
\newcommandtwoopt*{\rvXh}[2][][]{\optScript[{#1}][{#2}]{\hat{\xi}}}
\newcommandtwoopt*{\rvXtd}[2][][]{\optScript[{#1}][{#2}]{\tilde{\xi}}}
\newcommandtwoopt*{\rvY}[2][][]{\optScript[{#1}][{#2}]{\eta}}
\newcommandtwoopt*{\rvYtd}[2][][]{\optScript[{#1}][{#2}]{\tilde{\eta}}}
\newcommandtwoopt*{\rvYb}[2][][]{\optScript[{#1}][{#2}]{\bar{\eta}}}
\newcommandtwoopt*{\rvYh}[2][][]{\optScript[{#1}][{#2}]{\hat{\eta}}}
\newcommandtwoopt*{\rvYc}[2][][]{\optScript[{#1}][{#2}]{\check{\eta}}}
\newcommandtwoopt*{\rvZ}[2][][]{\optScript[{#1}][{#2}]{\zeta}}
\newcommandtwoopt*{\rvZd}[2][][]{\optScript[{#1}][{#2}]{\acAst{\zeta}}}
\newcommandtwoopt*{\rvZdd}[2][][]{\optScript[{#1}][{#2}]{\acAAst{\zeta}}}

\newcommandtwoopt*{\rvSpc}[2][][]{\optScript[{#1}][{#2}]{\vartheta}}

\newcommandtwoopt*{\valA}[2][][]{\optScript[{#1}][{#2}]{a}}
\newcommandtwoopt*{\valAlpha}[2][][]{\optScript[{#1}][{#2}]{\alpha}}

\newcommandtwoopt*{\valAtt}[2][][]{\ensuremath{\mathtt{\optScript[{[#1]}][{#2}]{a}}}}

\newcommandtwoopt*{\valB}[2][][]{\optScript[{#1}][{#2}]{b}}
\newcommandtwoopt*{\valBeta}[2][][]{\optScript[{#1}][{#2}]{\beta}}
\newcommandtwoopt*{\valC}[2][][]{\optScript[{#1}][{#2}]{c}}
\newcommandtwoopt*{\valO}[2][][]{\optScript[{#1}][{#2}]{\omega}}
\newcommandtwoopt*{\valX}[2][][]{\optScript[{#1}][{#2}]{x}}
\newcommandtwoopt*{\valXtd}[2][][]{\optScript[{#1}][{#2}]{\tilde{x}}}
\newcommandtwoopt*{\valXh}[2][][]{\optScript[{#1}][{#2}]{\hat{x}}}
\newcommandtwoopt*{\valXc}[2][][]{\optScript[{#1}][{#2}]{\check{x}}}
\newcommandtwoopt*{\valU}[2][][]{\optScript[{#1}][{#2}]{u}}
\newcommandtwoopt*{\valUh}[2][][]{\optScript[{#1}][{#2}]{\hat{u}}}
\newcommandtwoopt*{\valUtd}[2][][]{\optScript[{#1}][{#2}]{\tilde{u}}}
\newcommandtwoopt*{\valV}[2][][]{\optScript[{#1}][{#2}]{v}}
\newcommandtwoopt*{\valVh}[2][][]{\optScript[{#1}][{#2}]{\hat{v}}}
\newcommandtwoopt*{\valW}[2][][]{\optScript[{#1}][{#2}]{w}}
\newcommandtwoopt*{\valY}[2][][]{\optScript[{#1}][{#2}]{y}}
\newcommandtwoopt*{\valZ}[2][][]{\optScript[{#1}][{#2}]{z}}
\newcommandtwoopt*{\valZd}[2][][]{\optScript[{#1}][{#2}]{\acAst{z}}}
\newcommandtwoopt*{\valZdd}[2][][]{\optScript[{#1}][{#2}]{\acAAst{z}}}

\newcommand{\mInf}[3][]{\ensuremath{I#1(#2;#3#1)}}

\newcommand{\ESy}{\mathrm{E}}

\newcommand{\iDn}{\mathrm{i}}

\newcommand{\funcV}{\mathrm{V}}

\newcommandtwoopt*{\E}[3][][]{%
\ifthenelse{\equal{#2}{}}%
{\ensuremath{\ESy#1(#3#1)}}%
{\ensuremath{\ESy#1[#3#1]}}%
}

\newcommand{\cor}[3][]{\ensuremath{\mathrm{cor}#1(#2,#3#1)}}

\newcommand{\var}[2][]{\ensuremath{\mathrm{var}#1(#2#1)}}
\newcommand{\dx}{\ensuremath{\,\mathrm{d}}}

\newcommandtwoopt*{\dVec}[2][][]{\optScript[{#1}][{#2}]{d}}
\newcommandtwoopt*{\sMat}[2][][]{\optScript[{#1}][{#2}]{S}}
\newcommandtwoopt*{\kMat}[2][][]{\optScript[{#1}][{#2}]{K}}

\ifCLASSOPTIONpeerreview

\else

\fi
\makeatletter
\newlength \figwidth
\if@twocolumn
\else
\fi
\makeatother

\renewcommand\figurename{Figure}

\def\thefootnote{}
    
\begin{document}
\title{On the Distribution of the Information Density of Gaussian Random Vectors: Explicit Formulas and Tight Approximations\\[0.5ex]}

\author{
\authorblockN{ Jonathan Huffmann$^*$ and Martin Mittelbach$^{**}$}\\[1.5ex]
\authorblockA{$^{*}$\small Chair of Theoretical Information Technology \\ Department of Electrical and Computer Engineering\\ Technical University of Munich,  Germany\\
$^{**}$\small Chair of Communications Theory, Communications Laboratory \\ Department of Electrical Engineering and Information Technology\\ Dresden University of Technology,  Germany\\ 
Email: jonathan.huffmann@tum.de, martin.mittelbach@tu-dresden.de}
}

\maketitle


\begin{abstract}\footnote{ $^{**}$\,Funded in part by the German Research Foundation (DFG, Deutsche Forschungsgemeinschaft) as part of Germany's Excellence Strategy -- EXC 2050/1 -- Project ID 390696704 -- Cluster of Excellence "{}Centre for Tactile Internet with Human-in-the-Loop"{} (CeTI) of Technische Universit\"at Dresden.} 
Based on the canonical correlation analysis we derive series representations of the probability density function (PDF) and the cumulative distribution function (CDF) of  the information density of arbitrary Gaussian random vectors as well as a general formula to calculate the central moments. 
Using the general results we give closed-form expressions of the PDF and CDF and explicit formulas of the central moments for important special cases. Furthermore, we derive recurrence formulas and tight approximations of the general series representations,   
 which allow very efficient numerical calculations with an arbitrarily high accuracy as demonstrated with an implementation in \textsc{Python} publicly available on \textsc{GitLab}. 
Finally, we discuss the (in)validity of Gaussian approximations of the information density.

\end{abstract}


\setcounter{footnote}{0}
\renewcommand*{\thefootnote}{\arabic{footnote}}

\IEEEpubidadjcol

\section{Introduction and Main Results}

Let  \rvX\ and \rvY\ be arbitrary random variables on an abstract probability space $(\Omega,\mathcal{F},\meP)$ such that the joint distribution \meP[\rvX\rvY]\ is absolutely continuous \wrt\ the product $\meP[\rvX]\otimes\meP[\rvY]$ of the marginal distributions \meP[\rvX]\ and \meP[\rvY]. If  $\frac{\dx\meP[\rvX\rvY]}{\dx\meP[\rvX]\otimes\meP[\rvY]}$ denotes the Radon-Nikodym derivative of \meP[\rvX\rvY]\ \wrt\ $\meP[\rvX]\otimes\meP[\rvY]$, then 
\begin{align*}
\iDn(\rvX;\rvY)=\log\bigg(\frac{\dx\meP[\rvX\rvY]}{\dx\meP[\rvX]\otimes\meP[\rvY]}(\rvX,\rvY)\bigg)
\end{align*}
is called the information density of \rvX\ and \rvY. 
The expectation $\E{\iDn(\rvX;\rvY)}=\mInf{\rvX}{\rvY}$  of the information density, called mutual information, plays a key role in characterizing the asymptotic channel coding performance in terms of channel capacity.
The non-asymptotic performance, however, is determined by the higher-order moments of the information density and its probability distribution. 
Achievability and converse bounds that allow a finite blocklength analysis  of the optimum channel coding rate are closely related to the distribution function of the information density, also called information spectrum by Han \cite{Han2003}. 
 Moreover,  based on the variance of the information density tight second-order finite blocklength approximations of the optimum code rate can be derived for various important channel models.   
First work on a non-asymptotic information theoretic analysis was already published in the early years of information theory by Shannon \cite{Shannon1959}, Dobrushin \cite{Dobrushin1961}, and Strassen \cite{Strassen1964}, among others. 
Due to the seminal work  of Polyanskiy et al.\ \cite{Polyanskiy2010} considerable progress has been made  in this area. The results of Polyanskiy on the one hand and the requirements of current and future wireless networks regarding  latency and reliability on the other hand stimulated a significant new interest in this type of analysis (Durisi et al.\ \cite{Durisi2016}).

The information density $\iDn(\rvX;\rvY)$ in the case when $\rvX$ and $\rvY$ are jointly Gaussian is of special interest due to the prominent role of the Gaussian distribution. 
Let $\rvX=(\rvX[1],\rvX[2],\ldots,\rvX[{p}])$ and $\rvY=(\rvY[1],\rvY[2],\ldots,\rvY[{q}])$ be real-valued random vectors\footnote{For notational convenience we write vectors as row vectors. However, in expressions where matrix or vector multiplications occur, we consider all vectors as column vectors.}
 with nonsingular  covariance matrices $R_{\rvX}$ and $R_{\rvY}$ and cross-covariance matrix $R_{\rvX\rvY}$ with rank $r=\rank(R_{\xi\eta})$. 
Without loss of generality for the subsequent results, we assume the expectation of all random variables to be zero. 
If $(\rvX[1],\rvX[2],\ldots,\rvX[p],\rvY[1],\rvY[2],\ldots,\rvY[q])$ is a Gaussian random vector, then Pinsker \cite[Ch.\,9.6]{Pinsker1964} has shown that  the distribution of the information density  $\iDn(\rvX;\rvY)$  coincides with the distribution of the random variable 
\begin{align}\label{EQ:SUM-REPRESENTATION-OF-INFODENSITY}
\rvIdnsI &=\frac{1}{2}\sum_{i=1}^r\ccaC[i]\big(\rvXtd[i][2]-\rvYtd[i][2]\big)+\mInf{\rvX}{\rvY}.
\end{align}
In this representation  $\rvXtd[1], \rvXtd[2],\ldots,\rvXtd[r], \rvYtd[1], \rvYtd[2], $\ldots$, \rvYtd[r]$ are independent and identically distributed (i.i.d.)  Gaussian random variables  with zero mean and unit variance, $\ccaC[1]\geq\ccaC[2]\geq\ldots\geq\ccaC[r]>0$ denote the positive canonical correlations (see \Cref{PROPOSITION:CCA}) of \rvX\ and \rvY\ in descending order, and the mutual information \mInf{\rvX}{\rvY} has the form 
\begin{align}\label{EQ:SUM-REPRESENTATION-OF-MUTUAL-INFO}
\mInf{\rvX}{\rvY}&=\frac{1}{2}\sum_{i=1}^{r}\log\bigg(\frac{1}{1-\ccaC[i][2]}\bigg).
\end{align}
The rank $r$ of the cross-covariance matrix $R_{\rvX\rvY}$ satisfies $0\leq r \leq \min\{p,q\}$ and for $r=0$ we have $\iDn(\rvX;\rvY)\equiv 0$ almost surely and $\mInf{\rvX}{\rvY}=0$. This corresponds to $\meP[\rvX\rvY]=\meP[\rvX]\otimes\meP[\rvY]$ and the independence of \rvX\ and \rvY\ such that the resulting information density is deterministic. Throughout the rest of the paper we exclude this degenerated case when the information density is considered and assume subsequently the setting and notation introduced above with $r \geq 1$.  
As customary notation we further write  $\R$, $\No$, and $\N$ to denote the set of real numbers, non-negative integers, and positive integers, respectively.  

\subsection{Main theorems}

Based on \eqref{EQ:SUM-REPRESENTATION-OF-INFODENSITY} we derive series representations of the probability density function (PDF) and the cumulative distribution function (CDF) of the information density $\iDn(\rvX;\rvY)$ given subsequently in \Cref{thm:pdfinf,thm:cdfinf}.  
These representations are useful as they allow tight approximations with  errors as low as desired by finite sums as shown in \Cref{SEQ:FINITE-SUM-APPROXIMATIONS}. 
Moreover, the recurrence formulas derived in \Cref{SEC:RECURSIVE-REPRESENTATION} allow very efficient numerical calculations.  

\begin{theorem}[PDF of information density]
  \label{thm:pdfinf}
The PDF $f_{\iDn(\xi;\eta)}$ of the information density $\iDn(\rvX;\rvY)$ is given by 
  \begin{multline}\label{EQ:PDF-INFO-DENSITY}
    f_{\iDn(\rvX;\rvY)}(x)=\frac{1}{\ccaC[r]\sqrt{\pi}}\sum_{k_{1}=0}^{\infty}\sum_{k_{2}=0}^{\infty}\dots
    \sum_{k_{r-1}=0}^{\infty}\left[\prod_{i=1}^{r-1}
    \frac{\ccaC[r]}{\ccaC[i]}\frac{(2k_{i})!}{(k_{i}!)^{2}4^{k_{i}}}
      \left(1-\frac{\ccaC[r][2]}{\ccaC[i][2]}\right)^{k_{i}}\right]
    \times\\
    \frac{\mathrm{K}_{\frac{r-1}{2}+k_{1}+k_{2}+\dots+k_{r-1}}
      \left(\left|\frac{x-I(\xi;\eta)}{\ccaC[r]}\right|\right)}
         {\Gamma\left(\frac{r}{2}+k_{1}+k_{2}+\dots+k_{r-1}\right)}
         \left|\frac{x-I(\xi;\eta)}{2\ccaC[r]}\right|^{\left(\frac{r-1}{2}+k_{1}+k_{2}+\dots+k_{r-1}\right)},
         \qquad x\in\R\backslash\{I(\xi;\eta)\},
          \end{multline}
	where $\Gamma(\cdot)$ denotes the gamma function	 
	\cite[Sec.\,5.2.1]{Olver2010} 
	and 	
	$\mathrm{K}_{\alpha}(\cdot)$ denotes the modified Bessel function of second kind and order $\alpha$ 
	\cite[Sec.\,10.25(ii)]{Olver2010}. 
	If $r \geq 2$ then $f_{\iDn(\xi;\eta)}(x)$ is also well defined for ${x=I(\xi;\eta)}$.
\end{theorem}

The method to obtain the result in \Cref{thm:pdfinf} is adopted from Mathai \cite{Mathai1982}, where a series  representation of the PDF of the sum of independent gamma distributed random variables is derived. Previous work of  Grad and Solomon \cite{Grad1955} and Kotz et.\ al.\ \cite{Kotz1967} goes in a similar direction as Mathai \cite{Mathai1982}, however it is not directly applicable since only the restriction to positive series coefficients is considered there. Using \Cref{thm:pdfinf} the series representation of the CDF of the information density in \Cref{thm:cdfinf} below is obtained.  The details of the derivations of \Cref{thm:pdfinf,thm:cdfinf}  are provided in \Cref{SEC:PROOFS-OF-MAIN-RESULTS}.

 \begin{theorem}[CDF of information density]
    \label{thm:cdfinf}
		The CDF $F_{\iDn(\xi;\eta)}$ of the information density $\iDn(\rvX;\rvY)$ is given by
    \begin{equation*}%
      F_{\iDn(\xi;\eta)}(x)=
      \begin{dcases}
        \rule{0ex}{3.5ex}\;\frac{1}{2}-\funcV\left(I(\xi;\eta)-x\right)&\text{if}\quad x \leq I(\xi;\eta)\\
        \;\frac{1}{2}+\funcV\left(x-I(\xi;\eta)\right)&\text{if}\quad x > I(\xi;\eta)\\[1ex]
      \end{dcases},
    \end{equation*}
    with $\funcV(z)$ defined by
    \begin{align}\nonumber
      \funcV(z)=&\sum_{k_{1}=0}^{\infty}\sum_{k_{2}=0}^{\infty}\dots
      \sum_{k_{r-1}=0}^{\infty}\left[\prod_{i=1}^{r-1}
        \frac{\ccaC[r]}{\ccaC[i]}\frac{(2k_{i})!}{(k_{i}!)^{2}4^{k_{i}}}
        \left(1-\frac{\ccaC[r][2]}{\ccaC[i][2]}\right)^{k_{i}}\right]
      \frac{z}{2\ccaC[r]}
      \times\\ \nonumber 
      &\bigg[
        \mathrm{K}_{\frac{r-1}{2}+k_{1}+k_{2}+\dots+k_{r-1}}\left(\frac{z}{\ccaC[r]}\right)
        \mathrm{L}_{\frac{r-3}{2}+k_{1}+k_{2}+\dots+k_{r-1}}\left(\frac{z}{\ccaC[r]}\right)+\\ \label{EQ:CDF-INFO-DENSITY}
        &\mathrm{K}_{\frac{r-3}{2}+k_{1}+k_{2}+\dots+k_{r-1}}\left(\frac{z}{\ccaC[r]}\right)
        \mathrm{L}_{\frac{r-1}{2}+k_{1}+k_{2}+\dots+k_{r-1}}\left(\frac{z}{\ccaC[r]}\right)
        \bigg],\qquad\quad z\geq 0,
  \end{align}
where 
$\mathrm{L}_{\alpha}(\cdot)$ denotes the modified Struve $\mathrm{L}$ function of order $\alpha$ 
	\cite[Sec.\,11.2]{Olver2010}. 
\end{theorem}	

A general formula for the central moments of the information  density is given in the next \namecref{THM:MOMENTS-OF-INFODENSITY} and proved in \Cref{SEC:PROOFS-OF-MAIN-RESULTS}. 

\begin{theorem}[Central moments of information density]
  \label{THM:MOMENTS-OF-INFODENSITY}
The $m$-th central moment $\E[\big]{[\mathrm{i}(\rvX;\rvY)-\mInf{\rvX}{\rvY}]^m}$ of the information density $\mathrm{i}(\rvX;\rvY)$ is given by 
\begin{align}\label{EQ:CENTRAL-MOMENTS-GENRAL-CASE}
      \E[\big]{[\mathrm{i}(\rvX;\rvY)-\mInf{\rvX}{\rvY}]^m}=
      \begin{dcases}
        \sum_{(m_{1},m_{2},\cdots,m_{r})\in\mathcal{K}_{m,r}^{[2]}}m!\,\prod_{i=1}^{r}\frac{(2m_{i})!}{4^{m_i} (m_{i}!)^2} \,\ccaC[i]^{2m_{i}} & \text{ if }\; m=2\tilde{m}\\
				\,0  & \text{ if }\; m=2\tilde{m}-1
      \end{dcases},
\end{align}
for all $\tilde{m}\in\N$, where $\mathcal{K}_{m,r}^{[2]}=\big\{(m_{1},m_{2},\dots, m_{r})\in\mathds{N}^{r}_{0}: 2m_{1}+2m_{2}+\cdots+2m_{r}=m\big\}$. 

\end{theorem}

Pinsker \cite[Eq.\,(9.6.17)]{Pinsker1964} provided a formula for 
$\sum_{i=1}^r\E[\big]{\big[\frac{\ccaC[i]}{2}(\rvXtd[i][2]-\rvYtd[i][2])\big]^m}$, which he called "{}derived $m$-th central moment"{}  of the information density,  where $\rvXtd[i]$ and $\rvYtd[i]$ are given as in \eqref{EQ:SUM-REPRESENTATION-OF-INFODENSITY}. These special moments  coincide for $m=2$ with the usual central moments considered in \Cref{THM:MOMENTS-OF-INFODENSITY}.

\subsection{Special cases}

A simple but important special case for which the series representations in   \Cref{thm:pdfinf,thm:cdfinf} simplify to a single summand and the  sum of products in \Cref{THM:MOMENTS-OF-INFODENSITY} simplifies to a single product is  considered in the following \namecref{COR:PDF-CDF-EQUAL-CORRELATIONS}.

\newpage

 \begin{corollary}[PDF, CDF, and central moments of information density for equal canonical correlations]\label{COR:PDF-CDF-EQUAL-CORRELATIONS}%
		If all canonical correlations are equal	
    \begin{equation*}
      \ccaC[1]=\ccaC[2]=\ldots=\ccaC[r], 
    \end{equation*}    
		then we have the following simplified results.
		
		\begin{inparaenum}[(i)]
		\item\label{COR:PDF-CDF-EQUAL-CORRELATIONS-PART-I} The PDF $f_{\iDn(\xi;\eta)}$ of the information density $\iDn(\rvX;\rvY)$ simplifies to   
    \begin{equation}\label{EQ:PDF-INFO-DENSITY-EQUAL-CCA}
      f_{\iDn(\xi;\eta)}(x)=\frac{1}{\ccaC[r]\sqrt{\pi}\Gamma\left(\frac{r}{2}\right)}
      \mathrm{K}_{\frac{r-1}{2}}
      \left(\left|\frac{x-I(\xi;\eta)}{\ccaC[r]}\right|\right)
         \left|\frac{x-I(\xi;\eta)}{2\ccaC[r]}\right|^{\frac{r-1}{2}},
         \qquad x\in\R\backslash\{I(\xi;\eta)\},
    \end{equation}
    where $I(\xi;\eta)$ is given by
    \begin{equation*}
      I(\xi;\eta)=-\frac{r}{2}\log\left(1-\ccaC[r][2]\right).
    \end{equation*}
			If $r \geq 2$ then $f_{\iDn(\xi;\eta)}(x)$ is also well defined for ${x=I(\xi;\eta)}$. 

		\item\label{COR:PDF-CDF-EQUAL-CORRELATIONS-PART-II} The CDF $F_{\iDn(\xi;\eta)}$ of the information density $\iDn(\rvX;\rvY)$  is given by 
		    \begin{equation}\label{EQ:CDF-INFO-DENSITY-EQUAL-CCA}
      F_{\iDn(\xi;\eta)}(x)=%
      \begin{dcases}
        \rule{0ex}{3.5ex}\;\frac{1}{2}-\funcV\left(I(\xi;\eta)-x\right)&\text{if}\quad x \leq I(\xi;\eta)\\
        \;\frac{1}{2}+\funcV\left(x-I(\xi;\eta)\right)&\text{if}\quad x > I(\xi;\eta)\\[1ex]
      \end{dcases},
    \end{equation}
    with $\funcV(z)$ defined by
    \begin{equation}\label{EQ:CDF-INFO-DENSITY-EQUAL-CCA-FUNCTION-V}
      \funcV(z)=\frac{z}{2\ccaC[r]}
      \bigg[
        \mathrm{K}_{\frac{r-1}{2}}\left(\frac{z}{\ccaC[r]}\right)
        \mathrm{L}_{\frac{r-3}{2}}\left(\frac{z}{\ccaC[r]}\right)+
        \mathrm{K}_{\frac{r-3}{2}}\left(\frac{z}{\ccaC[r]}\right)
        \mathrm{L}_{\frac{r-1}{2}}\left(\frac{z}{\ccaC[r]}\right)
        \bigg],\qquad z\geq 0.
  \end{equation}    

\item\label{COR:PDF-CDF-EQUAL-CORRELATIONS-PART-III} The $m$-th central moment $\E[\big]{[\mathrm{i}(\rvX;\rvY)-\mInf{\rvX}{\rvY}]^m}$ of the information density $\iDn(\rvX;\rvY)$ has the  form
\begin{align*}
      \E[\big]{[\mathrm{i}(\rvX;\rvY)-\mInf{\rvX}{\rvY}]^m}=
      \begin{dcases}
        \frac{m!}{\big(\nicefrac{m}{2}\big)!}\,\left(\prod_{j=1}^{\nicefrac{m}{2}}\bigg(\frac{r}{2}+j-1\bigg)\right)\ccaC[r][m] & \text{ if }\; m=2\tilde{m}\\
				\,0  & \text{ if }\; m=2\tilde{m}-1
      \end{dcases},
\end{align*}
for all $\tilde{m}\in\N$. 

\end{inparaenum} 
\end{corollary}

Clearly, if all canonical correlations are equal, then the only nonzero term in the series \eqref{EQ:PDF-INFO-DENSITY} and \eqref{EQ:CDF-INFO-DENSITY} occur for $k_1=k_2=\ldots=k_{r-1}=0$.  For this single summand the product in squared brackets in \eqref{EQ:PDF-INFO-DENSITY} and \eqref{EQ:CDF-INFO-DENSITY} is equal to $1$ by applying $0^0=1$, which yields the results of part \eqref{COR:PDF-CDF-EQUAL-CORRELATIONS-PART-I} and \eqref{COR:PDF-CDF-EQUAL-CORRELATIONS-PART-II} in \Cref{COR:PDF-CDF-EQUAL-CORRELATIONS}. Details of the  derivation of part \eqref{COR:PDF-CDF-EQUAL-CORRELATIONS-PART-III} of the \namecref{COR:PDF-CDF-EQUAL-CORRELATIONS} are provided in \Cref{SEC:PROOFS-OF-MAIN-RESULTS}. 

 Note, if all canonical correlations are equal then we can rewrite 
\eqref{EQ:SUM-REPRESENTATION-OF-INFODENSITY} as follows 
\begin{align*}
\rvIdnsI &=\frac{\ccaC[r]}{2}\left(\sum_{i=1}^r\rvXtd[i][2]-\sum_{i=1}^r\rvYtd[i][2]\right)+\mInf{\rvX}{\rvY}. 
\end{align*}
This implies that $\rvIdnsI$ coincides with the distribution of the random variable  
\begin{align*}
\rvIdnsI[*] &= \frac{\ccaC[r]}{2}\big(\rvZ[1]-\rvZ[2]\big) + \mInf{\rvX}{\rvY}, 
\end{align*}
where $\rvZ[1]$ and $\rvZ[2]$ are i.i.d.\ $\chi^2$-distributed random variables with $r$ degrees of freedom. With this representation we can obtain the expression of the PDF given in \eqref{EQ:PDF-INFO-DENSITY-EQUAL-CCA} also from \cite[Sec.\,4.A.4]{Simon2006}.

\begin{numpar}[Special cases of \Cref{COR:PDF-CDF-EQUAL-CORRELATIONS}]
The case when all canonical correlations are equal is important because it occurs in various situations. The subsequent cases follow from the properties of canonical correlations given in \Cref{PROPOSITION:CCA}. 

\begin{inparaenum}[(i)]

\item Assume that 
\begin{gather}\label{EQ:INDEPENDENT-AND-EQUAL-CORRELATIONS-I}
\cor{\rvX[i]}{\rvY[i]}=\rho \neq 0,\qquad i=1,2,\ldots,k\leq\min\{p,q\}\\ 
\label{EQ:INDEPENDENT-AND-EQUAL-CORRELATIONS-II}
\cor{\rvX[i]}{\rvY[i]}=0,\qquad i=k+1,\ldots,\min\{p,q\},\\\label{EQ:INDEPENDENT-AND-EQUAL-CORRELATIONS-III}
\cor{\rvX[i]}{\rvX[j]}=0,\quad\cor{\rvY[i]}{\rvY[j]}=0,\quad\cor{\rvX[i]}{\rvY[j]}=0,\qquad i\neq j,
\end{gather}
where $\cor{\cdot}{\cdot}$ denotes the Pearson correlation coefficient.   
Then $r=k$ and $\ccaC[i]=|\rho|$ for all $i=1,2,\ldots,r$. 
Note, if $p=q=k$ then for \eqref{EQ:INDEPENDENT-AND-EQUAL-CORRELATIONS-I}--\eqref{EQ:INDEPENDENT-AND-EQUAL-CORRELATIONS-III} to hold it is sufficient that the two-dimensional random vectors $(\rvX[i],\rvY[i])$ are i.i.d. However, the identical distribution of the $(\rvX[i],\rvY[i])$'s is not necessary.
In Laneman \cite{Laneman2006} the distribution of the information density for an  
additive white Gaussian noise channel with i.i.d.\ Gaussian input is determined.   This is a special case of the case with i.i.d.\ random vectors $(\rvX[i],\rvY[i])$ just mentioned. 
In Wu and Jindal \cite{WuJindal2011} and in Buckingham and Valenti \cite{BuckinghamValenti2008} an approximation of the information density by a Gaussian random variable is considered for the setting in \cite{Laneman2006}. 
A special case very similar to that in \cite{Laneman2006} is also considered in Polyanskiy et al.\ \cite[Sec.\,III.J]{Polyanskiy2010}.
To the best of the authors knowledge explicit formulas for the general case as considered in this paper are not available yet in the literature. 

\item Assume that  \eqref{EQ:INDEPENDENT-AND-EQUAL-CORRELATIONS-I}--\eqref{EQ:INDEPENDENT-AND-EQUAL-CORRELATIONS-III} are satisfied. Further assume that $\hat{A}$ is a  real nonsingular matrix of dimension $p\times p$ and $\hat{B}$ is a  real nonsingular matrix of dimension $q\times q$. Then the random vectors
\begin{align*}
\rvXh=\hat{A}\,\rvX\qquad\text{ and }\qquad \rvYh=\hat{B}\,\rvY
\end{align*}
have the same canonical correlations as the random vectors \rvX\  and \rvY, \ie, $\ccaC[i]=|\rho|$ for all $i=1,2,\ldots,k\leq\min\{p,q\}$.

 \item If $r = 1$, \ie, if the cross-covariance matrix $R_{\rvX,\rvY}$ has rank $1$, then \Cref{COR:PDF-CDF-EQUAL-CORRELATIONS} obviously applies.   
Clearly, the most simple special case with $r=1$ occurs for $p=q=1$, where $\ccaC[1]=|\cor{\rvX[1]}{\rvY[1]}|$. %

As a simple multivariate example let the covariance matrix of $(\rvX[1],\rvX[2],\ldots,\rvX[p],\rvY[1],\rvY[2],\ldots,\rvY[q])$ be given by the Kac-Murdock-Szeg\"o matrix %
\begin{align*}
\begin{pmatrix}
R_{\rvX} & R_{\rvX\rvY}\\
R_{\rvX\rvY} & R_{\rvY}
\end{pmatrix}
=
\Big(\rho^{|i-j|}\Big)_{i,j=1}^{p+q},
\end{align*}
which is related to the covariance function of a 
first-order autoregressive process, where $0<|\rho|<1$. Then $r=\rank(R_{\rvX\rvY})=1$ and $\ccaC[1]=|\rho|$.

\item As yet another example assume $p=q$ and $R_{\rvX\rvY}=\rho R_{\rvX}^{\frac{1}{2}}R_{\rvY}^{\frac{1}{2}}$ for some $0<|\rho|<1$. Then $\ccaC[i]=|\rho|$  for $i=1,2,\ldots,r=q$. 
Here, $A^{\frac{1}{2}}$ denotes the square root of the real-valued positive semidefinite matrix $A$, i.\,e., the unique positive semidefinite matrix $B$  such that $BB = A$.

\end{inparaenum}

\end{numpar}

\begin{numpar}[More on special cases with simplified formulas] Let us evaluate the general formulas given above for some relevant parameter values.  

\begin{inparaenum}[(i)]
\item \textit{Single canonical correlation coefficient.} In the most simple case there is only a single non-zero  canonical correlation coefficient, \ie, $r=1$.   (Recall, in the beginning of the paper we have excluded the degenerated case when all canonical correlations are zero.)  
Then  the formulas of the PDF and the $m$-th central moment in  \Cref{COR:PDF-CDF-EQUAL-CORRELATIONS} simplify to the form 
    \begin{equation*}
      f_{\iDn(\xi;\eta)}(x)=\frac{1}{\ccaC[1]\pi}
      \mathrm{K}_{0}
      \left(\left|\frac{x-I(\xi;\eta)}{\ccaC[1]}\right|\right),
         \qquad x\in\mathds{R}\backslash\{I(\xi;\eta)\},
    \end{equation*}
and 
\begin{align}\label{EQ:CENTRAL-MOMENT-SINGLE-CCA}
      \E[\big]{[\mathrm{i}(\rvX;\rvY)-\mInf{\rvX}{\rvY}]^m}=
      \begin{dcases}
        \bigg(\frac{m!}{\big(\nicefrac{m}{2}\big)!}\bigg)^2\left(\frac{\ccaC[1]}{2}\right)^m & \text{ if }\; m=2\tilde{m}\\
				\,0  & \text{ if }\; m=2\tilde{m}-1
      \end{dcases},
\end{align}
for all $\tilde{m}\in\N$.  
A  formula equivalent to \eqref{EQ:CENTRAL-MOMENT-SINGLE-CCA} is also provided by  
 Pinsker \cite[Lemma~9.6.1]{Pinsker1964} who considered the special case $p=q=1$, which implies $r=1$.   

\item \textit{Second and forth central moment.} To demonstrate how the general formula given in \Cref{THM:MOMENTS-OF-INFODENSITY} is used we first consider $m=2$. In this case the summation indices $m_1,m_2,\ldots,m_r$ have to satisfy  $m_i=1$ for a single $i\in\{1,2,\ldots,r\}$ whereas the remaining $m_i$'s have to be zero. Thus \eqref{EQ:CENTRAL-MOMENTS-GENRAL-CASE} evaluates for $m=2$ to 
\begin{align}\label{EQ:VARIANCE-OF-INFO-DENSITY}
\E[\big]{[\mathrm{i}(\rvX;\rvY)-\mInf{\rvX}{\rvY}]^2} = \var[\big]{\mathrm{i}(\rvX;\rvY)}=\sum_{i=1}^r\ccaC[i][2].
\end{align}
As a slightly more complex example let $m=4$. In this case either we  have 
$m_i=2$ for a single $i\in\{1,2,\ldots,r\}$ whereas the remaining $m_i$'s are  zero or we have $m_{i_1}=m_{i_2}=1$ for two $i_1 \neq i_2 \in \{1,2,\ldots,r\}$ whereas the remaining $m_i$'s have to be zero. 
Thus \eqref{EQ:CENTRAL-MOMENTS-GENRAL-CASE} evaluates for $m=4$ to 
\begin{align*}
\E[\big]{[\mathrm{i}(\rvX;\rvY)-\mInf{\rvX}{\rvY}]^4} = 9\sum_{i=1}^r\ccaC[i][4] + 6\sum_{i=2}^r\sum_{j=1}^{i-1} \ccaC[i][2]\ccaC[j][2]. 
\end{align*}

\item \textit{Even number of equal canonical correlations.} As in \Cref{COR:PDF-CDF-EQUAL-CORRELATIONS} assume that all canonical correlations are equal and additionally assume that the number $r$ of canonical correlations is even, \ie, $r=2\tilde{r}$ for some $\tilde{r}\in\N$. Then we can use \cite[Secs.\,10.47.9,\,10.49.1, and 10.49.12]{Olver2010}  to obtain the following relation for the modified  Bessel function $\mathrm{K}_{\alpha}(\cdot)$ of second kind and order $\alpha$
\begin{align}\label{EQ:MODIFIED-BESSEL-FOR-HALF-INTEGERS}
\mathrm{K}_{\frac{r-1}{2}}
      (y)=\sqrt{\frac{\pi}{2}}\exp\left(-y\right)\sum_{i=0}^{\nicefrac{r}{2}-1}\frac{\big(\nicefrac{r}{2}-1+i\big)!}{\big(\nicefrac{r}{2}-1-i\big)!\,i!\,2^i}\,y^{-(i+\frac{1}{2})}, \qquad y\in(0,\infty). 
\end{align} 
Plugging \eqref{EQ:MODIFIED-BESSEL-FOR-HALF-INTEGERS} into \eqref{EQ:PDF-INFO-DENSITY-EQUAL-CCA} and rearranging terms yields the following expression for  the PDF of the information density.
\begin{equation*}
      f_{\iDn(\xi;\eta)}(x)=\frac{1}{\ccaC[r]2^{r-1}\big(\nicefrac{r}{2}-1\big)!}\exp\left(-\left|\frac{x-\mInf{\rvX}{\rvY}}{\ccaC[r]}\right|\right)\sum_{i=0}^{\nicefrac{r}{2}-1}\frac{\big(2(\nicefrac{r}{2}-1)-i\big)!\,2^i}{\big(\nicefrac{r}{2}-1-i\big)!\,i!}\left|\frac{x-\mInf{\rvX}{\rvY}}{\ccaC[r]}\right|^i,
         \quad x\in\R
\end{equation*}
By integration we obtain for the function $\funcV(\cdot)$ in \eqref{EQ:CDF-INFO-DENSITY-EQUAL-CCA-FUNCTION-V} the expression
    \begin{equation*}
		 \funcV(z) = \frac{1}{2}-\frac{1}{2^{r-1}\big(\nicefrac{r}{2}-1\big)!}\exp\left(-\frac{z}{\ccaC[r]}\right)\sum_{i=0}^{\nicefrac{r}{2}-1}\frac{\big(2(\nicefrac{r}{2}-1)-i\big)!\,2^i}{\big(\nicefrac{r}{2}-1-i\big)!}\sum_{j=0}^{i}\frac{1}{(i-j)!}\left(\frac{z}{\ccaC[r]}\right)^{i-j}, 
         \quad z\geq 0.
    \end{equation*}
Note, these special formulas can also be obtained directly from the results given in \cite[Sec.\,4.A.3]{Simon2006}. 
For the specific value $r=2$ we have 
    \begin{alignat*}{2}
      f_{\iDn(\xi;\eta)}(x) &=\frac{1}{2\ccaC[r]}\exp\left(-\left|\frac{x-\mInf{\rvX}{\rvY}}{\ccaC[r]}\right|\right),
         \quad &&x\in\R,\\
		 \funcV(z) &= \frac{1}{2}\left(1-\exp\left(-\frac{z}{\ccaC[r]}\right)\right),
         \quad &&z\geq 0
    \end{alignat*}
and for $r=4$ we have 
    \begin{alignat*}{2}
      f_{\iDn(\xi;\eta)}(x) &=\frac{1}{4\ccaC[r]}\exp\left(-\left|\frac{x-\mInf{\rvX}{\rvY}}{\ccaC[r]}\right|\right)\left(1+\left|\frac{x-\mInf{\rvX}{\rvY}}{\ccaC[r]}\right|\right),
         \quad &&x\in\R,\\
		 \funcV(z) &= \frac{1}{2}\left(1-\exp\left(-\frac{z}{\ccaC[r]}\right)\left(1+\frac{z}{2\ccaC[r]}\right)\right),
         \quad &&z\geq 0, 
    \end{alignat*}
illustrating the principal behavior of the PDF and CDF of the information density for equal canonical correlations.

\end{inparaenum}
\end{numpar}

The rest of the paper is organized as follows. In \Cref{SEC:PRELIMINARIES} we provide some background on the canonical correlation analysis and its application to the calculation of the information density and mutual information for Gaussian random vectors. Furthermore, \Cref{SEC:PRELIMINARIES}  contains  
auxiliary results required for the proofs of  \Cref{thm:pdfinf,thm:cdfinf,THM:MOMENTS-OF-INFODENSITY} given in \Cref{SEC:PROOFS-OF-MAIN-RESULTS}.  
Finite sum approximations and uniform bounds of the approximation error as well as recurrence formulas, which allow efficient and accurate numerical calculations of the PDF and CDF of the information density 
are derived in \Cref{SEC:NUMERICAL-APPROXIMATION}. Some examples and illustrations are provided in \Cref{SEC:EXAMPLES-AND-ILLUSTRATIONS}, where also  the (in)validity of Gaussian approximations is discussed.  
Finally, \Cref{SECTION:CONLUSIONS} summarizes the paper. 

\section{Preliminaries} 
\label{SEC:PRELIMINARIES}

First introduced by Hotelling \cite{Hotelling1936}, the canonical correlation analysis is a widely used linear method in multivariate statistics  to determine the maximum correlations between two sets of random variables.  
It allows a particularly simple and useful representation of the mutual information and the information density of Gaussian random vectors in terms of the so-called canonical correlations. This representation was first obtained by Gelfand and Yaglom \cite{Gelfand1959} and further extended by Pinsker \cite[$§$\,9]{Pinsker1964}.
For the convenience of the reader we summarize the essence of the canonical correlation analysis in the subsequent \Cref{PROPOSITION:CCA} and demonstrate in \Cref{PAR:INFORMATION-AND-CCA}  how it is applied to derive the representations in \eqref{EQ:SUM-REPRESENTATION-OF-INFODENSITY} and \eqref{EQ:SUM-REPRESENTATION-OF-MUTUAL-INFO}. Furthermore, the characteristic function of the information density and useful properties of the modified Bessel function $\mathrm{K}_{\alpha}$ of second kind and order $\alpha$ are derived in \Cref{SEC:CF-AND-BESSEL}. 

\subsection{Canonical Correlation Analysis}
\label{PROPOSITION:CCA}

The formulation of the canonical correlation analysis given below is particularly suitable for implementations. 
The corresponding results are given without proof. Details and thorough discussions can be found, \eg, in  H\"ardle and Simar \cite{Haerdle2015}, Koch \cite{Koch2014} or Timm \cite{Timm2002}. 

Based on the nonsingular  covariance matrices $R_{\rvX}$ and $R_{\rvY}$ of the random vectors $\rvX=(\rvX[1],\rvX[2],\ldots,\rvX[p])$ and $\rvY=(\rvY[1],\rvY[2],\ldots,\rvY[q])$, and the cross-covariance matrix $R_{\rvX\rvY}$ with rank $r=\rank(R_{\xi\eta})$ satisfying $0\leq r\leq\min\{p,q\}$ define the matrix 
\begin{equation*}
   M=R_{\rvX}^{-\frac{1}{2}}R_{\rvX\rvY}R_{\rvY}^{-\frac{1}{2}},
\end{equation*}
where the inverse matrices $R_{\rvX}^{-\frac{1}{2}}=\left(R_{\rvX}^{\frac{1}{2}}\right)^{-1}$ and $R_{\rvY}^{-\frac{1}{2}}=\left(R_{\rvY}^{\frac{1}{2}}\right)^{-1}$ can be obtain from diagonalizing $R_{\rvX}$ and $R_{\rvY}$. 
Then the matrix $M$ has a singular value decomposition
\begin{equation*}
 M=UD\transpose{V},
\end{equation*}
where $\transpose{V}$ denotes the transpose of $V$. The only non-zero entries $d_{1,1},d_{2,2},\ldots,d_{r,r}>0$ of the matrix $D= \big(d_{i,j}\big)_{i,j=1}^{p,q}$ are  called canonical correlations  of \rvX\ and \rvY, 
denoted by $\ccaC[i]=d_{i,i},i=1,2,\ldots,r$. The singular value decomposition can be chosen such that $\ccaC[1]\geq\ccaC[2]\geq\ldots\geq\ccaC[r]$ holds, which is assumed throughout the paper. 

Define the random vectors  $\rvXh=(\rvXh[1],\rvXh[2],\ldots,\rvXh[p])$ and $\rvYh=(\rvYh[1],\rvYh[2],\ldots,\rvYh[q])$ by
\begin{align*}
\rvXh=A\,\rvX\qquad\text{ and }\qquad\rvYh=B\,\rvY,
\end{align*}
where the nonsingular matrices $A$ and $B$ are given by
\begin{align*}
A=\transpose{U}R_{\rvX}^{-\frac{1}{2}}\qquad\text{ and }\qquad B=\transpose{V}R_{\rvY}^{-\frac{1}{2}}.
\end{align*}
Then the random variables $\rvXh[i],\rvYh[j]$ have the following correlation  properties
\begin{gather}\label{EQ:CCA-CORRELATIONS-I}
\cor{\rvXh[i]}{\rvYh[i]}=\ccaC[i],\qquad i=1,2,\ldots,r,\\\label{EQ:CCA-CORRELATIONS-II}
\cor{\rvXh[i]}{\rvYh[i]}=0,\qquad i=r+1,\ldots,\min\{p,q\},\\\label{EQ:CCA-CORRELATIONS-III}
\cor{\rvXh[i]}{\rvXh[j]}=\cor{\rvYh[i]}{\rvYh[j]}=\cor{\rvXh[i]}{\rvYh[j]}=0,\qquad i \neq j, 
\end{gather}
and the variances are all equal to $1$
\begin{align*}
\var{\rvXh[i]}=1,\qquad i=1,2,\ldots,p,\qquad\var{\rvYh[j]}=1,\qquad j=1,2,\ldots,q.
\end{align*}

\subsection{Mutual Information and Information Density in Terms of Canonical Correlations}
\label{PAR:INFORMATION-AND-CCA}

Based on the results of \Cref{PROPOSITION:CCA} we obtain for the mutual information and the information density
\begin{alignat*}{2}
\mInf{\rvX}{\rvY}&=\mInf{A\rvX}{B\rvY}&&=\mInf{\rvXh}{\rvYh}\\
\iDn(\rvX;\rvY)&=\iDn(A\rvX;B\rvY)&&=\iDn(\rvXh;\rvYh)\qquad\text{(\meP-almost surely)}
\end{alignat*}
because $A$ and $B$ are nonsingular matrices, which follows, \eg, from Pinsker \cite[Th.\,3.7.1]{Pinsker1964}. 

Since we consider the case where \rvX\ and \rvY\ are jointly Gaussian, \rvXh\ and \rvYh\ are jointly Gaussian as well. Therefore, the conditions in \eqref{EQ:CCA-CORRELATIONS-I}--\eqref{EQ:CCA-CORRELATIONS-III} imply that all random variables $\rvXh[i],\rvYh[j]$ are independent except for the pairs $(\rvXh[i],\rvYh[i])$, $i=1,2,\ldots,r$. 
This implies
\begin{align}\label{EQ:MUTUAL-INFO-CCA-SUM}
\mInf{\rvX}{\rvY}&=\sum_{i=1}^r\mInf{\rvXh[i]}{\rvYh[i]}\\\label{EQ:INFO-DENSITY-CCA-SUM}
\iDn(\rvX;\rvY)&=\sum_{i=1}^r \iDn(\rvXh[i];\rvYh[i])\qquad \text{(\meP-almost surely)}
\end{align}
where $\iDn(\rvXh[1];\rvYh[1]),\iDn(\rvXh[2];\rvYh[2]),\ldots,\iDn(\rvXh[r];\rvYh[r])$ are independent. The sum representations follow from the chain rules of mutual information and information density and the equivalence between independence and vanishing mutual information and information density.

Then \eqref{EQ:SUM-REPRESENTATION-OF-MUTUAL-INFO} is obtained from \eqref{EQ:MUTUAL-INFO-CCA-SUM}, \eqref{EQ:CCA-CORRELATIONS-I}, and the formula of mutual information for the bivariate Gaussian case. 
Since \rvXh[i]\ and \rvYh[i]\ are jointly Gaussian with zero mean, unit variance, and correlation $\cor{\rvXh[i]}{\rvYh[i]}=\ccaC[i]$ the information density $\iDn(\rvXh[i];\rvYh[i])$ is given by
\begin{align}\label{EQ:INFO-DENSITY-STD-GAUSSIAN}
\iDn(\rvXh[i];\rvYh[i])=-\frac{1}{2}\log(1-\ccaC[i][2])-\frac{\ccaC[i][2]}{2(1-\ccaC[i][2])}\bigg(\rvXh[i][2]-\frac{2\,\rvXh[i]\rvYh[i]}{\ccaC[i]}+\rvYh[i][2]\bigg),\qquad i=1,2,\ldots,r.
\end{align}
Now assume  $\rvXtd[1], \rvXtd[2],\ldots,\rvXtd[r], \rvYtd[1], \rvYtd[2], $\ldots$, \rvYtd[r]$ are i.i.d.\ Gaussian random variables  with zero mean and unit variance. Then for all $i=1,2,\ldots,r$ the distribution of the random vector
\begin{align*}
\frac{1}{\sqrt{2}}
\begin{pmatrix}
\sqrt{1+\ccaC[i]} & \sqrt{1-\ccaC[i]}\\
\sqrt{1+\ccaC[i]} & -\sqrt{1-\ccaC[i]}
\end{pmatrix}
\begin{pmatrix}
\rvXtd[i]\\
\rvYtd[i]
\end{pmatrix}
=\frac{1}{\sqrt{2}}\begin{pmatrix}
\sqrt{1+\ccaC[i]}\,\rvXtd[i]  +\sqrt{1-\ccaC[i]}\,\rvYtd[i]\\
\sqrt{1+\ccaC[i]}\,\rvXtd[i]  -\sqrt{1-\ccaC[i]}\,\rvYtd[i]
\end{pmatrix}
\end{align*}
coincides with the distribution of the random vector $(\rvXh[i],\rvYh[i])$. 
Plugging this into \eqref{EQ:INFO-DENSITY-STD-GAUSSIAN} we obtain together with 
\eqref{EQ:INFO-DENSITY-CCA-SUM} that the distribution of the information density $\iDn(\rvX;\rvY)$ coincides with the distribution of 
\eqref{EQ:SUM-REPRESENTATION-OF-INFODENSITY}.

\subsection{Characteristic Function and Modified Bessel Function}
\label{SEC:CF-AND-BESSEL}

To prove \Cref{thm:pdfinf} the following \namecref{LEMMA:CF-INFODENSITY} regarding  the characteristic function of the information density is utilized.  
The result of the \namecref{LEMMA:CF-INFODENSITY} is also used in  Ibragimov and Rozanov \cite{Ibragimov1970} but without proof. Therefore, the proof is given below for completeness.

\begin{lemma}[Characteristic function of (shifted) information density]\label{LEMMA:CF-INFODENSITY}
The characteristic function of the shifted information density $\iDn(\rvX;\rvY)-\mInf{\rvX}{\rvY}$ is equal to the characteristic function of the random variable 
  \begin{align}\label{EQ:SUM-REPRESENTATION-OF-INFODENSITY-SHIFTED}
    \rvIdnsII &=\frac{1}{2}\sum_{i=1}^r\ccaC[i]\big(\rvXtd[i][2]-\rvYtd[i][2]\big),
  \end{align}
  where $\rvXtd[1], \rvXtd[2],\ldots,\rvXtd[r], \rvYtd[1], \rvYtd[2], $\ldots$, \rvYtd[r]$ are  i.i.d.\ Gaussian random variables with zero mean and unit variance, and  $\ccaC[1],\ccaC[2],\ldots,\ccaC[r]$ are the canonical correlations of \rvX\ and \rvY. The  characteristic function of \rvIdnsII\ is given by 
\begin{equation}
	\label{eq:prodv}
        \varphi_{\rvIdnsII}(t)=\prod_{i=1}^{r}\frac{1}{\sqrt{1+\ccaC[i][2]t^{2}}},\qquad\qquad
        t\in\R. 
  \end{equation}
\end{lemma}

\begin{proof} Due to \eqref{EQ:SUM-REPRESENTATION-OF-INFODENSITY} the distribution of the shifted information density $\iDn(\rvX;\rvY)-\mInf{\rvX}{\rvY}$ coincides with the distribution of the random variable \rvIdnsII\ in \eqref{EQ:SUM-REPRESENTATION-OF-INFODENSITY-SHIFTED} such that the characteristic functions of $\iDn(\rvX;\rvY)-\mInf{\rvX}{\rvY}$ and \rvIdnsII\ are equal.   

	It is a well known fact that
  $\rvXtd[i][2]$ and $\rvYtd[i][2]$ in \eqref{EQ:SUM-REPRESENTATION-OF-INFODENSITY-SHIFTED} are chi-squared distributed random variables with one
  degree of freedom. Moreover, for the weighted random variables $\frac{\ccaC[i]}{2}\rvXtd[i][2]$ and
  $\frac{\ccaC[i]}{2}\rvYtd[i][2]$ we obtain the following relation for the probability distribution
  for a parameter $a\geq0$.
  \begin{equation}\label{EQ:GAMMA-TRAFO}
    \meP\bigg(\frac{\ccaC[i]}{2}\rvXtd[i][2]\leq a\bigg)=\meP\bigg(-\sqrt{2\ccaC[i][-1]a}
    \leq\rvXtd[i]\leq\sqrt{2\ccaC[i][-1]a}\,\bigg)
    =\frac{1}{\sqrt{2\pi}}\int\limits_{x=-\sqrt{2\ccaC[i][-1]a}}^{\sqrt{2\ccaC[i][-1]a}}\exp\bigg(-\frac{x^{2}}{2}\bigg)\dx x
  \end{equation}
  Substitution of $x=\sqrt{2\ccaC[i][-1]\tilde{x}}$ in  \eqref{EQ:GAMMA-TRAFO} yields together 
  with the symmetry of the integral 
  \begin{equation}\label{EQ:GAMMA-TRAFO-II}
    \meP\bigg(\frac{\ccaC[i]}{2}\rvXtd[i][2]\leq a\bigg)=\int\limits_{\tilde{x}=0}^{a}
    \frac{1}{\Gamma\left(\frac{1}{2}\right)\sqrt{\ccaC[i]\tilde{x}}}
    \exp\left(-\frac{\tilde{x}}{\ccaC[i]}\right)\dx\tilde{x}.
  \end{equation}
  From \eqref{EQ:GAMMA-TRAFO-II} the weighted random variables $\frac{\ccaC[i]}{2}\rvXtd[i][2]$ and
  $\frac{\ccaC[i]}{2}\rvYtd[i][2]$ are seen to be gamma distributed   with a scale
  parameter of $1/\ccaC[i]$ and shape parameter of $1/2$. The characteristic function of these
  random variables therefore admits the form
  \begin{equation*}
    \varphi_{\frac{\ccaC[i]}{2}\rvXtd[i][2]}(t)=\left(1-\ccaC[i]jt\right)^{-\frac{1}{2}}.
  \end{equation*}
  Furthermore, from the identity
  $\varphi_{-\frac{\ccaC[i]}{2}\rvXtd[i][2]}(t)=\varphi_{\frac{\ccaC[i]}{2}\rvXtd[i][2]}(-t)$ for
  the characteristic function and from the independence of $\rvXtd[i]$ and $\rvYtd[i]$ we obtain the
  characteristic function of $\rvIdnsII_{i}=\frac{\ccaC[i]}{2}(\rvXtd[i][2]-\rvYtd[i][2])$ to be
  given by
  \begin{equation*}
    \varphi_{\rvIdnsII_{i}}(t)=\left(1-\ccaC[i]jt\right)^{-\frac{1}{2}}\left(1+\ccaC[i]jt\right)^{-\frac{1}{2}}
    =\left(1+\ccaC[i][2]t^{2}\right)^{-\frac{1}{2}}.
  \end{equation*}
  Finally, because $\rvIdnsII$ in \eqref{EQ:SUM-REPRESENTATION-OF-INFODENSITY-SHIFTED} is given by the sum of the independent random variables
  $\rvIdnsII_{i}$ the characteristic function of $\rvIdnsII$ results from  multiplying the  individual characteristic functions of the random variables $\rvIdnsII_{i}$. By doing so we  obtain  \eqref{eq:prodv}, which completes the proof. 
\end{proof}

The subsequent properties of the modified Bessel function will be used to prove the main results. 

\begin{proposition}[Properties related to the function $\mathrm{K}_{\alpha}$]\label{PROP:PROPERTIES-OF-BESSEL-FUNCTION}
For all  $\alpha\in\R$ the function
\begin{align*}
y \mapsto y^{\alpha} \mathrm{K}_{\alpha}(y),\qquad y\in(0,\infty),
\end{align*}
where $\mathrm{K}_{\alpha}(\cdot)$ denotes the modified Bessel function of second kind and order $\alpha$ \cite[Sec.\,10.25(ii)]{Olver2010}, is strictly positive and strictly monotonically decreasing. Furthermore, if $\alpha>0$ then we have 
\begin{align}\label{EQ:PROP-SUPREMUM}
\lim_{y\rightarrow+0} y^{\alpha} \mathrm{K}_{\alpha}(y)=\sup_{y\in(0,\infty)} y^{\alpha} \mathrm{K}_{\alpha}(y)=\Gamma(\alpha)2^{\alpha-1}.
\end{align}
\end{proposition}

\begin{proof}
If $\alpha\in\R$ is fixed, then $\mathrm{K}_{\alpha}(y)$ is strictly positive and strictly monotonically decreasing \wrt\ $y\in(0,\infty)$ due to \cite[Secs.\,10.27.3 and 10.37]{Olver2010}. Furthermore, we obtain 
\begin{align*}
\frac{\dx  y^{\alpha} \mathrm{K}_{\alpha}(y)}{\dx y} = -y^{\alpha} \mathrm{K}_{\alpha-1}(y),\qquad y\in(0,\infty)
\end{align*}
by applying the rules to calculate derivatives of Bessel functions given in \cite[Sec.\,10.29(ii)]{Olver2010}. It follows that $y^{\alpha} \mathrm{K}_{\alpha}(y)$ is  strictly positive and strictly monotonically decreasing \wrt\ $y\in(0,\infty)$ for all fixed $\alpha\in\R$. 

Consider now the Basset integral formula as given in \cite[Sec.\,10.32.11]{Olver2010}
%
%
\begin{align}
	\label{eq:bassetint} 
  \mathrm{K_{\alpha}}(yz)=\frac{\Gamma\left(\alpha+\frac{1}{2}\right)(2z)^{\alpha}}{y^{\alpha}\sqrt{\pi}}
  \int\limits_{u=0}^{\infty}\frac{\cos(uy)}{\left(u^{2}+z^{2}\right)^{\alpha+\frac{1}{2}}}\dx u
  \qquad\text{for}\qquad|\arg(z)|<\frac{\pi}{2},\ y>0,\ \alpha>-\frac{1}{2}
\end{align}
and the integral
\begin{align}\label{EQ:RATIONAL-INTEGRAL} 
    \int\limits_{u=0}^{\infty}\frac{1}{\left(u^{2}+1\right)^{\alpha+\frac{1}{2}}}\dx u = \frac{\sqrt{\pi}\,\Gamma(\alpha)}{2\,\Gamma\left(\alpha+\frac{1}{2}\right)}\qquad\text{for}\qquad \alpha>0, 
  \end{align}
where the equality holds due to \cite[Secs.\,3.251.2 and 8.384.1]{Gradshteyn2007}. 
Using \eqref{eq:bassetint} and \eqref{EQ:RATIONAL-INTEGRAL} we obtain for all $\alpha>0$
\begin{align*}
\lim_{y\rightarrow+0} y^{\alpha} \mathrm{K}_{\alpha}(y) &= \lim_{y\rightarrow+0} 
\frac{\Gamma\left(\alpha+\frac{1}{2}\right)2^{\alpha}}{\sqrt{\pi}}
  \int\limits_{u=0}^{\infty}\frac{\cos(uy)}{\left(u^{2}+1\right)^{\alpha+\frac{1}{2}}}\dx u\\
	&= \frac{\Gamma\left(\alpha+\frac{1}{2}\right)2^{\alpha}}{\sqrt{\pi}}
  \int\limits_{u=0}^{\infty}\frac{\lim_{y\rightarrow+0}\cos(uy)}{\left(u^{2}+1\right)^{\alpha+\frac{1}{2}}}\dx u\\
		&= \frac{\Gamma\left(\alpha+\frac{1}{2}\right)2^{\alpha}}{\sqrt{\pi}}
  \int\limits_{u=0}^{\infty}\frac{1}{\left(u^{2}+1\right)^{\alpha+\frac{1}{2}}}\dx u\\
&=\Gamma(\alpha)2^{\alpha-1},
\end{align*}
where the second equality holds due to the dominated convergence theorem since $\cos(uy)/\left(u^{2}+1\right)^{\alpha+\frac{1}{2}}\leq 1/\left(u^{2}+1\right)^{\alpha+\frac{1}{2}}$.
Using the previously derived monotonicity we obtain \eqref{EQ:PROP-SUPREMUM}, which completes the proof. 
\end{proof}

\section{Proofs of Main Results} 
\label{SEC:PROOFS-OF-MAIN-RESULTS}

\subsection{Proof of \Cref{thm:pdfinf}} 

To prove \Cref{thm:pdfinf} we calculate the PDF $f_{\rvIdnsII}$ of the random variable \rvIdnsII\ introduced in \Cref{LEMMA:CF-INFODENSITY} by inverting 
the characteristic function $\varphi_{\rvIdnsII}$ given in \eqref{eq:prodv} via the integral
\begin{align}\label{EQ:CF-INVERSION}
f_{\rvIdnsII}(v)=\frac{1}{2\pi}\int_{-\infty}^{\infty}\varphi_{\rvIdnsII}(t)\exp\big(-\jmath t v\big)\dx t, \qquad v\in\R. 
\end{align}
Shifting the PDF of $\rvIdnsII$ by $\mInf{\rvX}{\rvY}$ we obtain the PDF $f_{\iDn(\rvX;\rvY)}$ of the
information density $\iDn(\rvX;\rvY)$ 
\begin{align}\label{EQ:SHIFTING-PDF}
f_{\iDn(\rvX;\rvY)}(x)=f_{\rvIdnsII}(x-\mInf{\rvX}{\rvY}),\qquad x\in\R.
\end{align}

The method used subsequently is based on the work of Mathai \cite{Mathai1982}. To invert the characteristic function $\varphi_{\rvIdnsII}$ we expand the factors in \eqref{eq:prodv} as
\begin{align}\label{EQ:EXPAND-FACTORS-OF-CF-I}
  \left(1+\ccaC[i]^{2}t^{2}\right)^{-\frac{1}{2}}
  &=\left(1+\ccaC[r]^{2}t^{2}\right)^{-\frac{1}{2}}
  \frac{\ccaC[r]}{\ccaC[i]}
  \left(1+\left(\frac{\ccaC[r]^{2}}{\ccaC[i]^{2}}-1\right)
  \left(1+\ccaC[r]^{2}t^{2}\right)^{-1}\right)^{-\frac{1}{2}}\\\label{EQ:EXPAND-FACTORS-OF-CF-II}
  &=\left(1+\ccaC[r]^{2}t^{2}\right)^{-\frac{1}{2}}
  \sum_{k=0}^{\infty}(-1)^{k}\binom{-\nicefrac{1}{2}}{k}\frac{\ccaC[r]}{\ccaC[i]}
  \left(1-\frac{\ccaC[r]^{2}}{\ccaC[i]^{2}}\right)^{k}
  \left(1+\ccaC[r]^{2}t^{2}\right)^{-k}. 
\end{align}
 In \eqref{EQ:EXPAND-FACTORS-OF-CF-II} we have used the binomial series 
\begin{align}\label{EQ:DEF-BINOMIAL-SERIES}
(1+y)^a = \sum_{k=0}^{\infty}\binom{a}{k}y^k
\end{align}
where $a\in\R$. The series is absolutely convergent for $|y|<1$ and 
\begin{align}\label{EQ:DEF-GENERALIZES-BINOMIAL-SERIES}
\binom{a}{k} = \prod_{\ell=1}^{k}\frac{a-\ell+1}{\ell},\qquad k\in\N,
\end{align}
denotes the generalized binomial coefficient with $\binom{a}{0}=1$.  
Since
\begin{align}\label{EQ:CONDITION-FOR-ABSOLUT-CONVERGENCE}
 \left|\left(1-\frac{\ccaC[r]^{2}}{\ccaC[i]^{2}}\right)
  \left(1+\ccaC[r]^{2}t^{2}\right)^{-1}\right|<1
\end{align}
holds for all $t\in\R$ the series in \eqref{EQ:EXPAND-FACTORS-OF-CF-II} is absolutely convergent for all $t\in\R$. 
Using the expansion in \eqref{EQ:EXPAND-FACTORS-OF-CF-II} and the absolute convergence together with the identity 
\begin{align}\label{EQ:IDENTITY-FOR-BINOMIAL-COEFF}
\binom{-\nicefrac{1}{2}}{k}=\frac{(-1)^k(2k)!}{(k!)^2 4^k}
\end{align}
we can rewrite the characteristic function $\varphi_{\rvIdnsII}$ as
\begin{equation}  \label{eq:sumvii}
  \varphi_{\rvIdnsII}(t)=\sum_{k_{1}=0}^{\infty}\sum_{k_{2}=0}^{\infty}\dots
  \sum_{k_{r-1}=0}^{\infty} \left[\prod_{i=1}^{r-1}
    \frac{\ccaC[r]}{\ccaC[i]}\frac{(2k_{i})!}{(k_{i}!)^{2}4^{k_{i}}}
    \left(1-\frac{\ccaC[r]^{2}}{\ccaC[i]^{2}}\right)^{k_{i}}\right]
  \left(1+\ccaC[r]^{2}t^{2}\right)^{-\left(\frac{r}{2}+k_{1}+k_{2}+\dots+k_{r-1}\right)},\quad t\in\R.
\end{equation}
To obtain the PDF $f_{\rvIdnsII}$ we evaluate the inversion integral \eqref{EQ:CF-INVERSION} based on the series representation in \eqref{eq:sumvii}.  
Since every series in \eqref{eq:sumvii} is absolutely convergent we can exchange summation and integration. Let \(\beta=\frac{r}{2}+k_{1}+k_{2}+\dots+k_{r-1}\). Then by symmetry we have  for the integral of a summand
\begin{equation}\label{EQ:INTEGRAL-OF-SERIES-TERM}
  \int\limits_{t=-\infty}^{\infty}\frac{\exp\left(-\jmath tv\right)}{(1+\ccaC[r]^{2}t^{2})^{\beta}} \dx t= 2\int\limits_{t=0}^{\infty}\frac{\cos\left(tv\right)}{(1+\ccaC[r]^{2}t^{2})^{\beta}} \dx t = \frac{2}{\ccaC[r]}\int\limits_{u=0}^{\infty}\frac{\cos\left(uv/\ccaC[r]\right)}{(1+u^2)^{\beta}} \dx u,
\end{equation}
where the second equality is a result of the substitution $t=u/\ccaC[r]$. 
By setting $z=1$, $\alpha=\beta-\frac{1}{2}\geq 0$ and $y=v/\ccaC[r]$ in the Basset integral formula given in \eqref{eq:bassetint} in the proof of \Cref{PROP:PROPERTIES-OF-BESSEL-FUNCTION} and using the symmetry with respect to $v$ we can evaluate \eqref{EQ:INTEGRAL-OF-SERIES-TERM} to the following form. 
\begin{equation}\label{EQ:INTEGRAL-OF-SERIES-TERM-II}
  \int\limits_{t=-\infty}^{\infty}\frac{\exp\left(-\jmath tv\right)}{(1+\ccaC[r]^{2}t^{2})^{\beta}}\dx t=
  \frac{\sqrt{\pi}}{\Gamma\left(\beta\right)2^{\beta-\frac{3}{2}}\ccaC[r]^{\beta+\frac{1}{2}}}
  \mathrm{K}_{\beta-\frac{1}{2}}\left(\frac{|v|}{\ccaC[r]}\right)|v|^{\beta-\frac{1}{2}},
  \qquad  v\in\R\backslash\{0\}
\end{equation}
Combining \eqref{EQ:CF-INVERSION}, \eqref{eq:sumvii} and \eqref{EQ:INTEGRAL-OF-SERIES-TERM-II} yields
\begin{multline}\label{eq:shftinfpd}
  f_{\rvIdnsII}(v)=\frac{1}{2\sqrt{\pi}}\sum_{k_{1}=0}^{\infty}\sum_{k_{2}=0}^{\infty}\dots
  \sum_{k_{r-1}=0}^{\infty}\left[\prod_{i=1}^{r-1}
    \frac{\ccaC[r]}{\ccaC[i]}\frac{(2k_{i})!}{(k_{i}!)^{2}4^{k_{i}}}
    \left(1-\frac{\ccaC[r]^{2}}{\ccaC[i]^{2}}\right)^{k_{i}}\right]
  \times\\
  \frac{\mathrm{K}_{\frac{r-1}{2}+k_{1}+k_{2}+\dots+k_{r-1}}\left(\frac{|v|}{\ccaC[r]}\right)
    |v|^{\left(\frac{r-1}{2}+k_{1}+k_{2}+\dots+k_{r-1}\right)}}
       {\Gamma\left(\frac{r}{2}+k_{1}+k_{2}+\dots+k_{r-1}\right)2^{\left(\frac{r-3}{2}+k_{1}+k_{2}+\dots+k_{r-1}\right)}
         \ccaC[r]^{\left(\frac{r+1}{2}+k_{1}+k_{2}+\dots+k_{r-1}\right)}},\qquad v\in\R\backslash\{0\}.
\end{multline}
Slightly rearranging terms and applying \eqref{EQ:SHIFTING-PDF} yields the PDF of the information density $\iDn(\rvX;\rvY)$ given in \eqref{EQ:PDF-INFO-DENSITY}. 

It remains to show that $f_{\iDn(\rvX;\rvY)}(x)$ is also well defined for $x=\mInf{\rvX}{\rvY}$ if $r \geq 2$. 
Indeed, if $r\geq 2$ then we can use \Cref{PROP:PROPERTIES-OF-BESSEL-FUNCTION} to obtain 
\begin{align*}
\lim_{x\rightarrow \mInf{\rvX}{\rvY}}f_{\iDn(\rvX;\rvY)}(x)=%
\frac{1}{2\ccaC[r]\sqrt{\pi}}\sum_{k_{1}=0}^{\infty}\sum_{k_{2}=0}^{\infty}\dots
    \sum_{k_{r-1}=0}^{\infty}\left[\prod_{i=1}^{r-1}
    \frac{\ccaC[r]}{\ccaC[i]}\frac{(2k_{i})!}{(k_{i}!)^{2}4^{k_{i}}}
      \left(1-\frac{\ccaC[r][2]}{\ccaC[i][2]}\right)^{k_{i}}\right]
    \times\\
    \frac{\Gamma\left(\frac{r-1}{2}+k_{1}+k_{2}+\dots+k_{r-1}\right)}
         {\Gamma\left(\frac{r-1}{2}+k_{1}+k_{2}+\dots+k_{r-1}+\frac{1}{2}\right)}
         \end{align*}
where we used the exchangeability of the limit and the summation due to the absolute convergence of the series. 
Since $\Gamma(\alpha)/\Gamma(\alpha+\frac{1}{2})$ is decreasing \wrt\ $\alpha\geq\frac{1}{2}$, we have
\begin{align*}
\frac{\Gamma\left(\frac{r-1}{2}+k_{1}+k_{2}+\dots+k_{r-1}\right)}
         {\Gamma\left(\frac{r-1}{2}+k_{1}+k_{2}+\dots+k_{r-1}+\frac{1}{2}\right)}\leq\frac{\Gamma\left(\frac{r-1}{2}\right)}
         {\Gamma\left(\frac{r-1}{2}+\frac{1}{2}\right)}\leq\sqrt{\pi}.
\end{align*}
Then together with \eqref{EQ:SERIES-OF-PRODUCTS} in the proof of \Cref{THEOREM:APPROXIMATION-ERROR-PDF-CDF} it follows that $\lim_{x\rightarrow \mInf{\rvX}{\rvY}}f_{\iDn(\rvX;\rvY)}(x)$ exists and is
finite, which completes the proof. \hfill $\blacksquare$

\subsection{Proof of \Cref{thm:cdfinf}}
To prove \Cref{thm:cdfinf} we calculate the CDF $F_{\rvIdnsII}$ of the random variable \rvIdnsII\ introduced in \Cref{LEMMA:CF-INFODENSITY} by integrating the PDF $f_{\rvIdnsII}$ given in \eqref{eq:shftinfpd}. 
Shifting the CDF of $\rvIdnsII$ by $\mInf{\rvX}{\rvY}$ we obtain the CDF $F_{\iDn(\rvX;\rvY)}$ of the
information density $\iDn(\rvX;\rvY)$ 
\begin{align*}%
F_{\iDn(\rvX;\rvY)}(x)=F_{\rvIdnsII}(x-\mInf{\rvX}{\rvY}),\qquad x\in\R.
\end{align*}
Using the symmetry of $f_{\rvIdnsII}$ 
we can write  
\begin{equation*}%
  F_{\rvIdnsII}(z)=\mathrm{P}(\rvIdnsII\leq z)=
  \begin{dcases}
    \frac{1}{2}-\int_{v=0}^{-z}f_{\rvIdnsII}(v)\dx v&\text{for}\quad z\leq0\\
    \frac{1}{2}+\int_{v=0}^{z}f_{\rvIdnsII}(v)\dx v&\text{for}\quad z> 0
  \end{dcases}.
\end{equation*}
It is therefore sufficient to evaluate the integral 
\begin{align}\label{EQ:INTEGRAL-VZ}
  \funcV(z):=\int_{v=0}^{z}f_{\rvIdnsII}(v)\dx v \quad\text{for}\quad z\geq 0.
\end{align}
To calculate the integral \eqref{EQ:INTEGRAL-VZ}, we plug \eqref{eq:shftinfpd} into \eqref{EQ:INTEGRAL-VZ} and exchange integration and summation, which is justified by the monotone  convergence theorem. 
To evaluate the integral of a summand consider the following identity
\begin{multline} \label{eq:struveint}
  \int_{x=0}^{z}x^{\alpha}\mathrm{K}_{\alpha}(x)\dx x=
  2^{\alpha-1}\sqrt{\pi}\Gamma\left(\alpha+\frac{1}{2}\right)
  z\bigg[\mathrm{K}_{\alpha}(z)\mathrm{L}_{\alpha-1}(z)+
    \mathrm{K}_{\alpha-1}(z)\mathrm{L}_{\alpha}(z)\bigg]
  \quad\text{for}\quad\alpha >-\frac{1}{2}
\end{multline}
given in \cite[Sec.\,1.12.1.3]{Prudnikov1986a},  
where $\mathrm{L}_{\alpha}(\cdot)$ denotes the modified Struve $\mathrm{L}$ function of order $\alpha$ 
	\cite[Sec.\,11.2]{Olver2010}.  
Using \eqref{eq:struveint} 
with \(\alpha=\frac{r-1}{2}+k_{1}+k_{2}+\dots+k_{r-1}\geq 0\) we obtain
 \begin{align*}
  \funcV(z)=&\sum_{k_{1}=0}^{\infty}\sum_{k_{2}=0}^{\infty}\dots
  \sum_{k_{r-1}=0}^{\infty}\left[\prod_{i=1}^{r-1}
    \frac{\ccaC[r]}{\ccaC[i]}\frac{(2k_{i})!}{(k_{i}!)^{2}4^{k_{i}}}
    \left(1-\frac{\ccaC[r]^{2}}{\ccaC[i]^{2}}\right)^{k_{i}}\right]
  \frac{z}{2\ccaC[r]}
    \times\\
    &\bigg[
      \mathrm{K}_{\frac{r-1}{2}+k_{1}+k_{2}+\dots+k_{r-1}}\left(\frac{z}{\ccaC[r]}\right)
      \mathrm{L}_{\frac{r-3}{2}+k_{1}+k_{2}+\dots+k_{r-1}}\left(\frac{z}{\ccaC[r]}\right)+\\
     &\mathrm{K}_{\frac{r-3}{2}+k_{1}+k_{2}+\dots+k_{r-1}}\left(\frac{z}{\ccaC[r]}\right)
      \mathrm{L}_{\frac{r-1}{2}+k_{1}+k_{2}+\dots+k_{r-1}}\left(\frac{z}{\ccaC[r]}\right)
      \bigg]\qquad\text{for}\quad z\geq0,
  \end{align*}
	which completes the proof. \hfill $\blacksquare$

\subsection{Proof of \Cref{THM:MOMENTS-OF-INFODENSITY}}
Using the random variable
\begin{align*}
\rvIdnsII=\sum_{i=1}^r\rvIdnsII[i] \qquad \text{with} \qquad \rvIdnsII[i]=\frac{\ccaC[i]}{2}(\rvXtd[i]-\rvYtd[i])
\end{align*}
introduced in \Cref{LEMMA:CF-INFODENSITY}  and the well-known multinomial theorem \cite[Sec.\,26.4.9]{Olver2010} 
\begin{align*}
\big(y_1+y_2+\ldots y_r\big)^m=\sum_{(\ell_1,\ell_2,\ldots,\ell_r)\in\mathcal{K}_{m,r}}m!\,\prod_{i=1}^r\frac{y_i^{\ell_i}}{\ell_i!},
\end{align*}
where $\mathcal{K}_{m,r}=\big\{(\ell_{1},\ell_{2},\dots,\ell_{r})\in\mathds{N}^{r}_{0}: \ell_{1}+\ell_{2}+\cdots+\ell_{r}=m\big\}$,  
we can write the $m$-th central moment of the information density $\mathrm{i}(\rvX;\rvY)$ as 
\begin{align}\nonumber
\E[\Big]{[\mathrm{i}(\rvX;\rvY)-\mInf{\rvX}{\rvY}]^m} &=\E[\Bigg]{\bigg[\sum_{i=1}^r\rvIdnsII[i]\bigg]^m}\\\label{EQ:MTH-MOMENT-AS-PRODUCT}
&=\sum_{(\ell_1,\ell_2,\ldots,\ell_r)\in\mathcal{K}_{m,r}} m!\,\prod_{i=1}^r\frac{\E[\big]{\rvIdnsII[i][{\ell_i}]}}{\ell_i!}. 
\end{align}
To obtain \eqref{EQ:MTH-MOMENT-AS-PRODUCT} we have exchanged expectation and summation and additionally used the identity $\E[\big]{\prod_{i=1}^r \rvIdnsII[i][{\ell_i}]}=\prod_{i=1}^r\E[\big]{\rvIdnsII[i][{\ell_i}]}$, which holds due  to the independence of the random variables $\rvIdnsII[1],\rvIdnsII[2],\ldots,\rvIdnsII[r]$. 

Based on the relation between the $\ell$-th central moment of a random variable and the $\ell$-th derivative of its characteristic function at $0$ we further have
\begin{align}\label{EQ:LITH-MOMENT-OF-NUI}
\E[\big]{\rvIdnsII[i][{\ell_i}]} = (-\jmath)^{\ell_i}\frac{\mathrm{d}^{\ell_i}}{\mathrm{d}t^{\ell_i}}\varphi_{\rvIdnsII_{i}}(t)\bigg|_{t=0}, 
\end{align}
where 
\begin{equation*}
    \varphi_{\rvIdnsII_{i}}(t)=\left(1+\ccaC[i][2]t^{2}\right)^{-\frac{1}{2}},\qquad t\in\R, 
\end{equation*}
is the characteristic function  of the random variable $\rvIdnsII[i]$ derived in the proof of \Cref{LEMMA:CF-INFODENSITY}. 
Consider now the binomial series expansion 
\begin{align*}
\varphi_{\rvIdnsII_{i}}(t) =\left(1+\ccaC[i][2]t^{2}\right)^{-\frac{1}{2}} = \sum_{{m_i=0}}^\infty \binom{-\nicefrac{1}{2}}{m_i}\left(\ccaC[i] t\right)^{2{m_i}},
\end{align*}
given in \eqref{EQ:DEF-BINOMIAL-SERIES}  in the proof of \Cref{thm:pdfinf}, which is absolutely convergent for all $t<\ccaC[i][{-1}]$. Further, consider the Taylor series expansion of the characteristic function $\varphi_{\rvIdnsII_{i}}$ at the point $0$ 
\begin{align*}
\varphi_{\rvIdnsII_{i}}(t) = \sum_{\ell_i=0}^\infty \left(\frac{\mathrm{d}^{\ell_i}}{\mathrm{d}t^{\ell_i}}\varphi_{\rvIdnsII_{i}}(t)\bigg|_{t=0}\right)\frac{t^{\ell_i}}{\ell_i !} .
\end{align*}
Both series expansions must be identical in an open interval around $0$ such that we obtain by comparing the series coefficients 
\begin{align*}
\frac{\mathrm{d}^{\ell_i}}{\mathrm{d}t^{\ell_i}}\varphi_{\rvIdnsII_{i}}(t)\bigg|_{t=0}=
\begin{dcases}
\ell_i ! \binom{-\nicefrac{1}{2}}{\nicefrac{\ell_i}{2}}\ccaC[i][{\ell_i}]\quad&\text{ if}\quad \ell_i=2m_i\\
0\quad&\text{ if}\quad \ell_i=2m_i-1\\
\end{dcases}
\end{align*}
for all $m_i\in\N$. 
With this result \eqref{EQ:LITH-MOMENT-OF-NUI} evaluates to  
\begin{align}\label{EQ:MOMENTS-OF-NUI-TILDE}
\E[\big]{\rvIdnsII[i][{\ell_i}]} = \begin{dcases}
\frac{\big(\ell_i !\big)^2}{\big((\nicefrac{\ell_i}{2}) !\big)^2\,4^{\frac{\ell_i}{2}}} \ccaC[i][{\ell_i}]\quad&\text{ if}\quad \ell_i=2m_i\\
0\quad&\text{ if}\quad \ell_i=2m_i-1\\
\end{dcases} 
\end{align}
for all $m_i\in\N$, where we have additionally used the identity \eqref{EQ:IDENTITY-FOR-BINOMIAL-COEFF}.

From \eqref{EQ:MOMENTS-OF-NUI-TILDE} and \eqref{EQ:MTH-MOMENT-AS-PRODUCT} we now obtain 
\begin{align*}
\E[\Big]{[\mathrm{i}(\rvX;\rvY)-\mInf{\rvX}{\rvY}]^m} &=0
\end{align*}
for all $m=2\tilde{m}-1$ with $\tilde{m}\in\N$ because if $m$ is odd then for all $(\ell_{1},\ell_{2},\dots,\ell_{r})\in\mathcal{K}_{m,r}$ at least one of the $\ell_i$'s has to be odd. 
If  $m=2\tilde{m}$ with $\tilde{m}\in\N$ we obtain from \eqref{EQ:MOMENTS-OF-NUI-TILDE} and \eqref{EQ:MTH-MOMENT-AS-PRODUCT} 
\begin{align*}
\E[\Big]{[\mathrm{i}(\rvX;\rvY)-\mInf{\rvX}{\rvY}]^m} %
&=\sum_{(\ell_1,\ell_2,\ldots,\ell_r)\in\mathcal{K}_{m,r}} m!\,\prod_{i=1}^r\frac{1}{\ell_i!}\frac{\big(\ell_i !\big)^2}{\big((\nicefrac{\ell_i}{2}) !\big)^2\,4^{\frac{\ell_i}{2}}} \ccaC[i][{\ell_i}]\\
&=\sum_{(m_1,m_2,\ldots,m_r)\in\mathcal{K}_{m,r}^{[2]}} m!\,\prod_{i=1}^r\frac{(2m_i)!}{\big(m_i !\big)^2\,4^{m_i}} \ccaC[i][{2m_i}],
\end{align*}
which completes the proof. ~  \hfill $\blacksquare$

\subsection{Proof of part \eqref{COR:PDF-CDF-EQUAL-CORRELATIONS-PART-III} of \Cref{COR:PDF-CDF-EQUAL-CORRELATIONS}} 

Using the random variable $\rvIdnsII$ as in the proof of  
\Cref{THM:MOMENTS-OF-INFODENSITY} we can write the $m$-th central moment of the information density $\mathrm{i}(\rvX;\rvY)$ as 
\begin{align*}
\E[\Big]{[\mathrm{i}(\rvX;\rvY)-\mInf{\rvX}{\rvY}]^m} &= \E[\big]{\rvIdnsII[][m]}\\
&=(-\jmath)^{m}\frac{\mathrm{d}^{m}}{\mathrm{d}t^{m}}\varphi_{\rvIdnsII}(t)\bigg|_{t=0}, 
\end{align*}
where the characteristic function $\varphi_{\rvIdnsII}$ of $\rvIdnsII$ is given by
\begin{equation*}
    \varphi_{\rvIdnsII}(t)=\left(1+\ccaC[r][2]t^{2}\right)^{-\frac{r}{2}},\qquad t\in\R, 
\end{equation*}
due to \Cref{LEMMA:CF-INFODENSITY} and the equality of all canonical correlations. 
Using the binomial series and the Taylor series expansion as in the proof of \Cref{THM:MOMENTS-OF-INFODENSITY}  we obtain
\begin{align*}
\frac{\mathrm{d}^{m}}{\mathrm{d}t^{m}}\varphi_{\rvIdnsII}(t)\bigg|_{t=0}=
\begin{dcases}
m ! \binom{-\nicefrac{r}{2}}{\nicefrac{m}{2}}\ccaC[r][{m}]\quad&\text{ if}\quad m=2\tilde{m}\\
0\quad&\text{ if}\quad m=2\tilde{m}-1\\
\end{dcases}
\end{align*}
for all $\tilde{m}\in\N$. Collecting terms yields
\begin{align*}
      \E[\big]{[\mathrm{i}(\rvX;\rvY)-\mInf{\rvX}{\rvY}]^m}=%
      \begin{dcases}
        \frac{m!}{\big(\nicefrac{m}{2}\big)!}\,\left(\prod_{j=1}^{\nicefrac{m}{2}}\bigg(\frac{r}{2}+j-1\bigg)\right)\ccaC[r][m] & \text{ if }\; m=2\tilde{m}\\
				\,0  & \text{ if }\; m=2\tilde{m}-1
      \end{dcases},
\end{align*}
where we have additionally used the definition of the generalized binomial coefficient given in \eqref{EQ:DEF-GENERALIZES-BINOMIAL-SERIES}. This completes the proof. 
~  \hfill $\blacksquare$

\section{Approximations and Recurrence Formulas} 
\label{SEC:NUMERICAL-APPROXIMATION}

In this \namecref{SEC:NUMERICAL-APPROXIMATION} we derive finite sum approximations of the PDF and CDF of the information density and uniform bounds of the approximation error. Furthermore, we derive recurrence formulas, which allow efficient numerical calculations with arbitrarily high accuracy as  demonstrated with an implementation in \textsc{Python} publicly available on \textsc{GitLab} \cite{HuffmannGitlab2021}. 

\subsection{Finite Sum Approximations}
\label{SEQ:FINITE-SUM-APPROXIMATIONS}
If there are at least two distinct canonical correlations, then the PDF $f_{\iDn(\rvX;\rvY)}$ and CDF $F_{\iDn(\rvX;\rvY)}$ of the information density $\iDn(\rvX;\rvY)$ are given by the infinite series in \Cref{thm:pdfinf,thm:cdfinf}. 
If we consider only a finite number of summands in these representations then we obtain approximations suitable in particular for numerical calculations.   
Let us consider for $r \geq 2$ and at least two distinct canonical correlations the following approximated PDF 
  \begin{multline}\label{EQ:APPROXIMATED-PDF}
    \hat{f}_{\iDn(\rvX;\rvY)}(x,n_1,n_2,\ldots,n_{r-1})=\frac{1}{\ccaC[r]\sqrt{\pi}}\sum_{k_{1}=0}^{n_1}\,\sum_{k_{2}=0}^{n_2}\dots
    \sum_{k_{r-1}=0}^{n_{r-1}}\left[\prod_{i=1}^{r-1}
    \frac{\ccaC[r]}{\ccaC[i]}\frac{(2k_{i})!}{(k_{i}!)^{2}4^{k_{i}}}
      \left(1-\frac{\ccaC[r][2]}{\ccaC[i][2]}\right)^{k_{i}}\right]
    \times\\
    \frac{\mathrm{K}_{\frac{r-1}{2}+k_{1}+k_{2}+\dots+k_{r-1}}
      \left(\left|\frac{x-I(\xi;\eta)}{\ccaC[r]}\right|\right)}
         {\Gamma\left(\frac{r}{2}+k_{1}+k_{2}+\dots+k_{r-1}\right)}
         \left|\frac{x-I(\xi;\eta)}{2\ccaC[r]}\right|^{\left(\frac{r-1}{2}+k_{1}+k_{2}+\dots+k_{r-1}\right)},
         \qquad x\in\R,
        \end{multline}
and CDF 
    \begin{equation}\label{EQ:APPROXIMATED-CDF}
      \hat{F}_{\iDn(\xi;\eta)}(x,n_1,n_2,\ldots,n_{r-1})=%
      \begin{dcases}
        \rule{0ex}{3.5ex}\;\frac{1}{2}-\hat{\funcV}\left(I(\xi;\eta)-x,n_1,n_2,\ldots,n_{r-1}\right)&\text{if}\quad x \leq I(\xi;\eta)\\
        \;\frac{1}{2}+\hat{\funcV}\left(x-I(\xi;\eta),n_1,n_2,\ldots,n_{r-1}\right)&\text{if}\quad x > I(\xi;\eta)\\[1ex]
      \end{dcases},
    \end{equation}
with $\hat{\funcV}\left(z,n_1,n_2,\ldots,n_{r-1}\right)$ defined by 
    \begin{align}\nonumber
      \hat{\funcV}\left(z,n_1,n_2,\ldots,n_{r-1}\right)=&\sum_{k_{1}=0}^{n_1}\,\sum_{k_{2}=0}^{n_2}\dots
      \sum_{k_{r-1}=0}^{n_{r-1}}\left[\prod_{i=1}^{r-1}
        \frac{\ccaC[r]}{\ccaC[i]}\frac{(2k_{i})!}{(k_{i}!)^{2}4^{k_{i}}}
        \left(1-\frac{\ccaC[r][2]}{\ccaC[i][2]}\right)^{k_{i}}\right]
      \frac{z}{2\ccaC[r]}
      \times\\ \nonumber 
      &\bigg[
        \mathrm{K}_{\frac{r-1}{2}+k_{1}+k_{2}+\dots+k_{r-1}}\left(\frac{z}{\ccaC[r]}\right)
        \mathrm{L}_{\frac{r-3}{2}+k_{1}+k_{2}+\dots+k_{r-1}}\left(\frac{z}{\ccaC[r]}\right)+\\ \label{EQ:APPROXIMATED-CDF-VZ}
        &\mathrm{K}_{\frac{r-3}{2}+k_{1}+k_{2}+\dots+k_{r-1}}\left(\frac{z}{\ccaC[r]}\right)
        \mathrm{L}_{\frac{r-1}{2}+k_{1}+k_{2}+\dots+k_{r-1}}\left(\frac{z}{\ccaC[r]}\right)
        \bigg],\qquad\quad z\geq 0, 
  \end{align}
where $n_1,n_2,\ldots n_{r-1}\in\N_0$ are constants specifying the number of summands taken into account for the approximation.  
The following \namecref{THEOREM:APPROXIMATION-ERROR-PDF-CDF} provides suitable error bounds  related to the approximation of the PDF and CDF by finite sums. 

\begin{theorem}[Bounds of the approximation error for PDF and CDF]\label{THEOREM:APPROXIMATION-ERROR-PDF-CDF}
Let $r \geq 2$ and assume that at least two canonical correlations are distinct. Then we have the following error bounds
\begin{multline}\label{EQ:ERROR-BOUND-PDF-OF-IDENS}
\big|{f}_{\iDn(\rvX;\rvY)}(x)-\hat{f}_{\iDn(\rvX;\rvY)}(x,n_1,n_2,\ldots,n_{r-1})\big|\leq\\
 \frac{\Gamma\left(\frac{r-1}{2}+n_1+n_2+\ldots+n_{r-1}\right)}{2\ccaC[r]\sqrt{\pi}\,\Gamma\left(\frac{r}{2}+n_1+n_2+\ldots+n_{r-1}\right)}\big(1-\hat{S}\left(n_1,n_2,\ldots,n_{r-1}\right)\big), \qquad x\in\R, 
\end{multline}
and 
\begin{align*}%
\big|\funcV(z)-\hat{\funcV}(z,n_1,n_2,\ldots,n_{r-1})\big|\leq
 \frac{1}{2}\big(1-\hat{S}\left(n_1,n_2,\ldots,n_{r-1}\right)\big), \qquad z\geq 0,
\end{align*}
where
\begin{align*}
\hat{S}\left(n_1,n_2,\ldots,n_{r-1}\right)=&\sum_{k_{1}=0}^{n_1}\,\sum_{k_{2}=0}^{n_2}\dots
      \sum_{k_{r-1}=0}^{n_{r-1}}\left[\prod_{i=1}^{r-1}
        \frac{\ccaC[r]}{\ccaC[i]}\frac{(2k_{i})!}{(k_{i}!)^{2}4^{k_{i}}}
        \left(1-\frac{\ccaC[r][2]}{\ccaC[i][2]}\right)^{k_{i}}\right].
\end{align*}
\end{theorem}

\begin{proof}
From the special case where all canonical correlations are equal we can conclude from the CDF given in \Cref{COR:PDF-CDF-EQUAL-CORRELATIONS} that the function
  \begin{align}\label{EQ:STRUVE-MCDONALD-FUNCTION}
    z\mapsto z
    \Big[
      \mathrm{K}_{\alpha}\left(z\right)
      \mathrm{L}_{\alpha-1}\left(z\right)+
      \mathrm{K}_{\alpha-1}\left(z\right)
      \mathrm{L}_{\alpha}\left(z\right)
      \Big],\qquad z \geq 0,
      \end{align}
is mononotically increasing for all $\alpha=(j-1)/2$, $j\in\N$ and that further 
  \begin{align}\label{EQ:LIMIT-STRUVE-MCDONALD}
    \lim_{z\to \infty}z
    \Big[
      \mathrm{K}_{\alpha}\left(z\right)
      \mathrm{L}_{\alpha-1}\left(z\right)+
      \mathrm{K}_{\alpha-1}\left(z\right)
      \mathrm{L}_{\alpha}\left(z\right)
      \Big]=1
   \end{align}
	holds. Using \eqref{EQ:LIMIT-STRUVE-MCDONALD} we obtain from \eqref{EQ:CDF-INFO-DENSITY}
  \begin{align*}
    \lim_{z\to \infty}2\funcV(z)=\sum_{k_{1}=0}^{\infty}\sum_{k_{2}=0}^{\infty}\dots
    \sum_{k_{r-1}=0}^{\infty}\left[\prod_{i=1}^{r-1}
      \frac{\ccaC[r]}{\ccaC[i]}\frac{(2k_{i})!}{(k_{i}!)^{2}4^{k_{i}}}
      \left(1-\frac{\ccaC[r][2]}{\ccaC[i][2]}\right)^{k_{i}}\right] 
   \end{align*}	
	by exchanging the limit and the summation, which is justified by the monotone convergence theorem. 
	Due to the properties of the CDF we have $\lim_{z\to \infty}2\funcV(z)=1$ and therefore 
	  \begin{align}\label{EQ:SERIES-OF-PRODUCTS}
    \sum_{k_{2}=0}^{\infty}\dots
    \sum_{k_{r-1}=0}^{\infty}\left[\prod_{i=1}^{r-1}
      \frac{\ccaC[r]}{\ccaC[i]}\frac{(2k_{i})!}{(k_{i}!)^{2}4^{k_{i}}}
      \left(1-\frac{\ccaC[r][2]}{\ccaC[i][2]}\right)^{k_{i}}\right] =1.
   \end{align}	
We now obtain
\begin{align*}
\big|\funcV(z)-\hat{\funcV}(z,n_1,n_2,\ldots,n_{r-1})\big|
 =\;
&\frac{1}{2}\sum_{(k_1,k_2,\ldots,k_{r-1})\in \mathcal{K}(n_1,n_2,\ldots,n_{r-1})}\left[\prod_{i=1}^{r-1}
        \frac{\ccaC[r]}{\ccaC[i]}\frac{(2k_{i})!}{(k_{i}!)^{2}4^{k_{i}}}
        \left(1-\frac{\ccaC[r][2]}{\ccaC[i][2]}\right)^{k_{i}}\right]
            \times\\ 
      &\frac{z}{\ccaC[r]}\bigg[
        \mathrm{K}_{\frac{r-1}{2}+k_{1}+k_{2}+\dots+k_{r-1}}\left(\frac{z}{\ccaC[r]}\right)
        \mathrm{L}_{\frac{r-3}{2}+k_{1}+k_{2}+\dots+k_{r-1}}\left(\frac{z}{\ccaC[r]}\right)\\
				&\quad+\mathrm{K}_{\frac{r-3}{2}+k_{1}+k_{2}+\dots+k_{r-1}}\left(\frac{z}{\ccaC[r]}\right)
        \mathrm{L}_{\frac{r-1}{2}+k_{1}+k_{2}+\dots+k_{r-1}}\left(\frac{z}{\ccaC[r]}\right)
        \bigg]\\
\leq\; & \frac{1}{2}\sum_{(k_1,k_2,\ldots,k_{r-1})\in \mathcal{K}(n_1,n_2,\ldots,n_{r-1})}\left[\prod_{i=1}^{r-1}
        \frac{\ccaC[r]}{\ccaC[i]}\frac{(2k_{i})!}{(k_{i}!)^{2}4^{k_{i}}}
        \left(1-\frac{\ccaC[r][2]}{\ccaC[i][2]}\right)^{k_{i}}\right]\\
				=\;&\frac{1}{2}\big(1-\hat{S}\left(n_1,n_2,\ldots,n_{r-1}\right)\big), 
\end{align*}
where $\mathcal{K}(n_1,n_2,\ldots,n_{r-1})=\N^{r-1}\setminus\{0,1,\ldots,n_1\}\times\{0,1,\ldots,n_2\}\times\ldots\times\{0,1,\ldots,n_{r-1}\}$. 
The inequality follows from the monotonicity of the function in \eqref{EQ:STRUVE-MCDONALD-FUNCTION}  and from \eqref{EQ:LIMIT-STRUVE-MCDONALD}. The last equality follows from \eqref{EQ:SERIES-OF-PRODUCTS}.
Similarly, we obtain 
\begin{align*}
&\hspace*{-1em}\big|{f}_{\iDn(\rvX;\rvY)}(x)-\hat{f}_{\iDn(\rvX;\rvY)}(x,n_1,n_2,\ldots,n_{r-1})\big|\\
=\;
&\frac{1}{\ccaC[r]\sqrt{\pi}}\sum_{(k_1,k_2,\ldots,k_{r-1})\in \mathcal{K}(n_1,n_2,\ldots,n_{r-1})}\left[\prod_{i=1}^{r-1}
        \frac{\ccaC[r]}{\ccaC[i]}\frac{(2k_{i})!}{(k_{i}!)^{2}4^{k_{i}}}
        \left(1-\frac{\ccaC[r][2]}{\ccaC[i][2]}\right)^{k_{i}}\right]
            \times\\ 
 &\hspace*{8em}\frac{\mathrm{K}_{\frac{r-1}{2}+k_{1}+k_{2}+\dots+k_{r-1}}
      \left(\left|\frac{x-I(\xi;\eta)}{\ccaC[r]}\right|\right)}
         {\Gamma\left(\frac{r}{2}+k_{1}+k_{2}+\dots+k_{r-1}\right)}
         \left|\frac{x-I(\xi;\eta)}{2\ccaC[r]}\right|^{\left(\frac{r-1}{2}+k_{1}+k_{2}+\dots+k_{r-1}\right)} \\[1ex]%
				\leq\; &  \frac{1}{\ccaC[r]\sqrt{\pi}}\sum_{(k_1,k_2,\ldots,k_{r-1})\in \mathcal{K}(n_1,n_2,\ldots,n_{r-1})}\left[\prod_{i=1}^{r-1}
        \frac{\ccaC[r]}{\ccaC[i]}\frac{(2k_{i})!}{(k_{i}!)^{2}4^{k_{i}}}
        \left(1-\frac{\ccaC[r][2]}{\ccaC[i][2]}\right)^{k_{i}}\right]\frac{\Gamma\left(\frac{r-1}{2}+k_1+k_2+\ldots+k_{r-1}\right)}{2\,\Gamma\left(\frac{r}{2}+k_1+k_2+\ldots+k_{r-1}\right)}\\[1ex]
\leq\; &\frac{\Gamma\left(\frac{r-1}{2}+n_1+n_2+\ldots+n_{r-1}\right)}{2\ccaC[r]\sqrt{\pi}\,\Gamma\left(\frac{r}{2}+n_1+n_2+\ldots+n_{r-1}\right)}\big(1-\hat{S}\left(n_1,n_2,\ldots,n_{r-1}\right)\big),
\end{align*}
where for the first inequality we have used \Cref{PROP:PROPERTIES-OF-BESSEL-FUNCTION} and for the second inequality we have used \eqref{EQ:SERIES-OF-PRODUCTS} and the decreasing monotonicity of $\Gamma(\alpha)/\Gamma(\alpha+\frac{1}{2})$  \wrt\ $\alpha\geq\frac{1}{2}$. This completes the proof. 
\end{proof}

\begin{remark}
Note that the bound in \eqref{EQ:ERROR-BOUND-PDF-OF-IDENS} can be further simplified using the inequality
\begin{align*}
\frac{\Gamma(\alpha)}{\Gamma\left(\alpha+\frac{1}{2}\right)}\leq\sqrt{\pi}. 
\end{align*}
Further note that the derived  error bounds are uniform in the sense that they only depend on the parameters of the given Gaussian distribution and the number of summands considered. As can be seen from \eqref{EQ:SERIES-OF-PRODUCTS} the bounds converge to zero as the number of summands jointly increase. 

\end{remark}

\subsection{Recurrence Formulas}
\label{SEC:RECURSIVE-REPRESENTATION}

A direct use of the formulas \eqref{EQ:APPROXIMATED-PDF}--\eqref{EQ:APPROXIMATED-CDF-VZ}  to numerically calculate the PDF and CDF of the information density   is rather inefficient since modified Bessel  and  Struve $\mathrm{L}$ functions have to be evaluated for every summand. 
Therefore, we derive subsequently 
recursive representations, which allow very efficient numerical calculations.  
Based on these recurrence relations an implementation in the programming language \textsc{Python} has been developed, which provides an efficient tool to numerically calculate the PDF and CDF of the information density with a predefined accuracy as high as desired. The developed source code as well as illustrating examples are made publicly available in an open access repository on \textsc{GitLab} \cite{HuffmannGitlab2021}.

The recursive approach developed below is based on the work   of Moschopoulos \cite{Moschopoulos1985}, which extended the work of Mathai \cite{Mathai1982}.
We adopt all the previous notation and assume $r \geq 2$ and at least two distinct canonical correlations (since otherwise we have the case of \Cref{COR:PDF-CDF-EQUAL-CORRELATIONS}, where the series reduce to a single summand).

First, we rewrite the series representations of the PDF and CDF  of the information density given in \Cref{thm:pdfinf} and \Cref{thm:cdfinf} in a form, which is suitable for recursive calculations. 
To begin with we define two functions appearing in the series representations \eqref{EQ:PDF-INFO-DENSITY} and \eqref{EQ:CDF-INFO-DENSITY} and  in the finite sum approximations \eqref{EQ:APPROXIMATED-PDF} and \eqref{EQ:APPROXIMATED-CDF-VZ}, which involve   the  modified Bessel function $\mathrm{K}_{\alpha}$ of second kind and order $\alpha$  and the modified  Struve $\mathrm{L}$ function $\mathrm{L}_{\alpha}$ of order $\alpha$. 
Let us define for all $k\in\No$ the functions $\mathrm{U}_{k}$ and $\mathrm{D}_{k}$ by
\begin{align}\label{EQ:DEFINITION-FUNCTION-UN}
  \mathrm{U}_{k}(z)=\frac{\mathrm{K}_{\frac{r-1}{2}+k}(z)}{\Gamma\left(\frac{r}{2}+k\right)}
  \left(\frac{z}{2}\right)^{\frac{r-1}{2}+k},\qquad z \geq 0  
\end{align}
and 
\begin{align}\label{EQ:DEFINITION-FUNCTION-DN}
  \mathrm{D}_{k}(z)=
      \frac{z}{2\ccaC[r]}
       \bigg[
        \mathrm{K}_{\frac{r-1}{2}+k}\left(\frac{z}{\ccaC[r]}\right)
        \mathrm{L}_{\frac{r-3}{2}+k}\left(\frac{z}{\ccaC[r]}\right)+
        \mathrm{K}_{\frac{r-3}{2}+k}\left(\frac{z}{\ccaC[r]}\right)
        \mathrm{L}_{\frac{r-1}{2}+k}\left(\frac{z}{\ccaC[r]}\right)
        \bigg],\qquad z \geq 0.
\end{align}
Furthermore, we define for all $k\in\No$ the coefficient $\delta_k$  by
\begin{align}\label{EQ:DEF-DELTAK}
  \delta_{k}=\sum_{(k_{1},k_{2},\dots,k_{r-1})\in\mathcal{K}_{k,r-1}}
  \left[%
	\prod_{i=1}^{r-1}
	\frac{(2k_{i})!}{(k_{i}!)^{2}4^{k_{i}}}
  \left(1-\frac{\ccaC[r][2]}{\ccaC[i][2]}\right)^{k_{i}}\right], 
\end{align}
where %
\begin{equation*}
\mathcal{K}_{k,r-1}=\big\{(k_{1},k_{2},\dots, k_{r-1})\in\mathds{N}^{r-1}_{0}: k_{1}+k_{2}+\cdots+k_{r-1}=k\big\}.
\end{equation*}
With these definitions we obtain the following alternative series representations of \eqref{EQ:PDF-INFO-DENSITY} and \eqref{EQ:CDF-INFO-DENSITY}  by observing that the multiple summations over the indices $k_1,k_2,\ldots,k_{r-1}$ can be shortened to one summation over the index $k=k_1+k_2+\ldots+k_{r-1}$. 
\begin{proposition}[Alternative representation of PDF and CDF of the information density]\label{PROP:ALT-SERIES-REPRESENTATIONS-OF-PDF-CDF} 
The PDF $f_{\iDn(\rvX;\rvY)}$ of the information density $\iDn(\rvX;\rvY)$ given in \Cref{thm:pdfinf} has the alternative series representation
\begin{align}\label{EQ:PDF-RECURSIVE}
f_{\iDn(\rvX;\rvY)}(x)=\frac{1}{\ccaC[r]\sqrt{\pi}}\left(\prod_{i=1}^{r-1}\frac{\ccaC[r]}{\ccaC[i]}\right)\sum_{k=0}^{\infty}\delta_{k}\mathrm{U}_{k}\left(\left|\frac{x-I(\xi;\eta)}{\ccaC[r]}\right|\right),\qquad x\in\R.
\end{align}
The function $\funcV(\cdot)$ specifying the CDF $F_{\iDn(\rvX;\rvY)}$ of the information density $\iDn(\rvX;\rvY)$ as given in \Cref{thm:cdfinf} has the alternative series  representation 
 \begin{align}\label{EQ:CDF-RECURSIVE-FUNCTION-V}
  \funcV(z)=\left(\prod_{i=1}^{r-1}\frac{\ccaC[r]}{\ccaC[i]}\right)\sum_{k=0}^{\infty}\delta_{k}\mathrm{D}_{k}(z),\qquad z\geq 0.
 \end{align}
\end{proposition}
Based on the representations in \Cref{PROP:ALT-SERIES-REPRESENTATIONS-OF-PDF-CDF} and with  recursive formulas for $\mathrm{U}_{k}(\cdot)$, $\mathrm{D}_{k}(\cdot)$ and $\delta_k$ we are in the position to calculate the PDF and CDF of the information density by a single summation over  completely recursively defined terms. In the following, we will derive recurrence relations for $\mathrm{U}_{k}(\cdot)$, $\mathrm{D}_{k}(\cdot)$ and $\delta_k$, which allow the desired efficient calculations.

\begin{lemma}[Recurrence formula of the function $\mathrm{U}_{k}$]\label{LEMMA:RECURRENCE-UN} If for all $k\in\No$ the function $\mathrm{U}_{k}$ is defined by \eqref{EQ:DEFINITION-FUNCTION-UN},   
then  $\mathrm{U}_{k}(z)$ satisfies for all $k \geq 2$ and $z\geq 0$ the recurrence formula
\begin{align}\label{EQ:RECURRENCE-FUNKTION-UN}
\mathrm{U}_{k}(z)=\frac{z^2}{(r+2k-2)(r+2k-4)}\,\mathrm{U}_{k-2}(z)+\frac{r+2k-3}{r+2k-2}\,\mathrm{U}_{k-1}(z).
\end{align}
\end{lemma}

\begin{proof}  
First assume $z=0$. Based on  \Cref{PROP:PROPERTIES-OF-BESSEL-FUNCTION}  we obtain for all $k\in\No$ 
\begin{align}\label{EQ:FUNCTION-UN-AT-O}
  \lim_{z\rightarrow + 0}\mathrm{U}_{k}(z)=\frac{\Gamma\left(\tfrac{r-1}{2}+k\right)}{2\,\Gamma\left(\tfrac{r}{2}+k\right)},
\end{align}
such that $\mathrm{U}_{k}(0)$ is well defined and finite.  
Using the recurrence relation $\Gamma(y+1)=y\Gamma(y)$ for the Gamma function \cite[Sec.\,8.331.1]{Gradshteyn2007} we have 
\begin{align*}
\frac{\Gamma\left(\tfrac{r-1}{2}+k\right)}{2\,\Gamma\left(\tfrac{r}{2}+k\right)}=
\frac{\left(\tfrac{r-1}{2}+k-1\right)}{\left(\tfrac{r}{2}+k-1\right)}\cdot\frac{\Gamma\left(\tfrac{r-1}{2}+k-1\right)}{2\,\Gamma\left(\tfrac{r}{2}+k-1\right)}.
\end{align*}
This shows together with \eqref{EQ:FUNCTION-UN-AT-O} that the recurrence formula \eqref{EQ:RECURRENCE-FUNKTION-UN} holds for $\mathrm{U}_{k}(0)$ and $k \geq 2$.  

Now assume $z>0$ and consider the recurrence formula
\begin{align}\label{EQ:BESSEL-RECURRENCE-FORMULA}
  z\mathrm{K}_{\alpha}(z)=z\mathrm{K}_{\alpha-2}(z)+2(\alpha-1)\mathrm{K}_{\alpha-1}(z)
\end{align}
for the modified Bessel function of the second kind and order $\alpha$  
\cite[Sec.\,8.486.10]{Gradshteyn2007}. 
%
%
Plugging \eqref{EQ:BESSEL-RECURRENCE-FORMULA} into \eqref{EQ:DEFINITION-FUNCTION-UN} for $\alpha=\frac{r-1}{2}+k$  yields for $k\geq 2$
\begin{align}\label{EQ:PRE-RECURRENCE-FUNKTION-UN}
  \mathrm{U}_{k}(z)=
  \frac{\mathrm{K}_{\frac{r-1}{2}+k-2}(z)}{\Gamma\left(\frac{r}{2}+k\right)}
  \left(\frac{z}{2}\right)^{\frac{r-1}{2}+k-2}\left(\frac{z}{2}\right)^{2}
  +\frac{\left(\frac{r-1}{2}+k-1\right)\mathrm{K}_{\frac{r-1}{2}+k-1}(z)}{\Gamma\left(\frac{r}{2}+k\right)}
  \left(\frac{z}{2}\right)^{\frac{r-1}{2}+k-1}. 
\end{align}
Using again the relation $\Gamma(y+1)=y\Gamma(y)$ we obtain  
%
\begin{align*}
 \Gamma\left(\tfrac{r}{2}+k\right)&=\left(\tfrac{r}{2}+k-1\right)\Gamma\left(\tfrac{r}{2}+k-1\right)\\
&=\left(\tfrac{r}{2}+k-1\right)\left(\tfrac{r}{2}+k-2\right)\Gamma\left(\tfrac{r}{2}+k-2\right),
\end{align*} 
which yields together with \eqref{EQ:PRE-RECURRENCE-FUNKTION-UN} and  \eqref{EQ:DEFINITION-FUNCTION-UN} the recurrence formula \eqref{EQ:RECURRENCE-FUNKTION-UN} for $\mathrm{U}_{k}(z)$ if $z>0$ and $k \geq 2$.   This completes the proof.  
\end{proof}

\begin{lemma}[Recurrence formula of the function $\mathrm{D}_{k}$]\label{LEMMA:RECURRENCE-DN} If for all $k\in\No$ the function $\mathrm{D}_{k}$ is defined by \eqref{EQ:DEFINITION-FUNCTION-DN},  
then  $\mathrm{D}_{k}(z)$ satisfies for all $k \geq 1$ and $z\geq 0$ the recurrence formula
\begin{align}\label{EQ:RECURRENCE-FUNKTION-DN}
\mathrm{D}_{k}(z)=\mathrm{D}_{k-1}(z)-\frac{1}{2\sqrt{\pi}\left(\frac{r}{2}+k-1\right)}\,\left(\frac{z}{\ccaC[r]}\right)\mathrm{U}_{k-1}\left(\frac{z}{\ccaC[r]}\right),
\end{align}
with $\mathrm{U}_{k}(\cdot)$ as defined in \eqref{EQ:DEFINITION-FUNCTION-UN}.  
\end{lemma}

\begin{proof} First assume $z=0$.  We have $\mathrm{D}_{k}(0)=0$ for all $k\in\No$ and from the proof of \Cref{LEMMA:RECURRENCE-UN} we have $\mathrm{U}_{k}(0)={\Gamma\left(\tfrac{r-1}{2}+k\right)}/{2\,\Gamma\left(\tfrac{r}{2}+k\right)}$ for all $k\in\No$ . Thus the left hand side and the right hand side of \eqref{EQ:RECURRENCE-FUNKTION-DN} are both zero, which shows that \eqref{EQ:RECURRENCE-FUNKTION-DN} holds for $z=0$ and $k\geq 1$. 

Now assume $z>0$ and consider the recurrence formula 
\begin{align*}
  z\mathrm{L}_{\alpha}(z)=z\mathrm{L}_{\alpha-2}(z)-2(\alpha-1)
  \mathrm{L}_{\alpha-1}(z)-\frac{2^{1-\alpha}z^{\alpha}}{\sqrt{\pi}\Gamma\left(\alpha+\frac{1}{2}\right)}
\end{align*}
for the modified Struve $\mathrm{L}$ function of order $\alpha$ \cite[Sec.\,11.4.25]{Olver2010}. Together with the recurrence formula \eqref{EQ:BESSEL-RECURRENCE-FORMULA} for the modified Bessel function of the second kind and order $\alpha$ 
we obtain
\begin{align}\label{EQ:RECURRENCE-STRUVE-L-BESSEL-K-COMB-I}
  z\mathrm{L}_{\alpha}(z)\mathrm{K}_{\alpha-1}(z) &=z\mathrm{L}_{\alpha-2}(z)\mathrm{K}_{\alpha-1}(z)-2(\alpha-1)
  \mathrm{L}_{\alpha-1}(z)\mathrm{K}_{\alpha-1}(z)-\frac{2^{1-\alpha}z^{\alpha}}{\sqrt{\pi}\Gamma\left(\alpha+\frac{1}{2}\right)}\mathrm{K}_{\alpha-1}(z),\\
	\label{EQ:RECURRENCE-STRUVE-L-BESSEL-K-COMB-II}
  z\mathrm{K}_{\alpha}(z)\mathrm{L}_{\alpha-1}(z) &=z\mathrm{K}_{\alpha-2}(z)\mathrm{L}_{\alpha-1}(z)+2(\alpha-1)\mathrm{K}_{\alpha-1}(z)\mathrm{L}_{\alpha-1}(z).
\end{align}
Plugging \eqref{EQ:RECURRENCE-STRUVE-L-BESSEL-K-COMB-I} and \eqref{EQ:RECURRENCE-STRUVE-L-BESSEL-K-COMB-II} into \eqref{EQ:DEFINITION-FUNCTION-DN} for $\alpha=\frac{r-1}{2}+k$ yields for  $k\geq 1$
\begin{align*}
  \mathrm{D}_{k}(z)=
      \frac{z}{2\ccaC[r]}
       \bigg[
        &\mathrm{K}_{\frac{r-1}{2}+k-1}\left(\frac{z}{\ccaC[r]}\right)
        \mathrm{L}_{\frac{r-3}{2}+k-1}\left(\frac{z}{\ccaC[r]}\right)+
        \mathrm{K}_{\frac{r-3}{2}+k-1}\left(\frac{z}{\ccaC[r]}\right)
        \mathrm{L}_{\frac{r-1}{2}+k-1}\left(\frac{z}{\ccaC[r]}\right)
        \bigg]\\
				&-\frac{1}{\sqrt{\pi}\,\Gamma\left(\frac{r}{2}+k\right)}\left(\frac{z}{2\ccaC[r]}\right)^{\frac{r-1}{2}+k}\mathrm{K}_{\frac{r-1}{2}+k-1}\left(\frac{z}{\ccaC[r]}\right).
\end{align*}
Together with \eqref{EQ:DEFINITION-FUNCTION-DN}, the identity $ \Gamma\left(\tfrac{r}{2}+k\right)=\left(\tfrac{r}{2}+k-1\right)\Gamma\left(\tfrac{r}{2}+k-1\right)$, and the definition of the function $\mathrm{U}_k(\cdot)$ in \eqref{EQ:DEFINITION-FUNCTION-UN} we obtain 
the recurrence formula \eqref{EQ:RECURRENCE-FUNKTION-DN} for $ \mathrm{D}_{k}(z)$ if $z>0$ and $k\geq 1$. This completes the proof. 
\end{proof}

\begin{lemma}[Recursive formula of the coefficient $\delta_k$]\label{LEMMA:RECURSIVE-FORMULA-DELTAK} 
The coefficient $\delta_{k}$ defined by \eqref{EQ:DEF-DELTAK} satisfies for all $k\in\No$ the recurrence formula
\begin{align}\label{EQ:DELTA-RECURRENCE}
\delta_{k+1}=\frac{1}{k+1}\sum_{j=1}^{k+1}j\,\gamma_{j}\,\delta_{k+1-j},%
\end{align}
where $\delta_0=1$ and
\begin{align}\label{EQ:GAMMAK-DEF}
  \gamma_{j}=\sum_{i=1}^{r-1}\frac{1}{2j}\left(1-\frac{\ccaC[r]^{2}}{\ccaC[i]^{2}}\right)^{j}. %
\end{align}
\end{lemma}

For the derivation of \Cref{LEMMA:RECURSIVE-FORMULA-DELTAK} we use an adapted  version of the method of Moschopoulos \cite{Moschopoulos1985} and the following auxiliary result.

\begin{lemma}\label{LEMMA:RECURRENCE-DEXP} For $k\in\No$ let  
$g$ be a real univariate $(k+1)$-times differentiable function. Then we have the following recurrence relation 
  for the $(k+1)$-th derivative of the composite function $h=\exp\left(g\right)$
  \begin{equation}\label{EQ:DEXP-RECURSIVE}
    h^{(k+1)}=\sum_{j=1}^{k+1}\binom{k}{j-1}g^{(j)}h^{(k-j+1)}, 
  \end{equation}
  where $f^{(i)}$ denotes the $i$-th derivative of the function $f$ with $f^{(0)}=f$.  
\end{lemma}
\begin{proof}
  We prove the assertion of \Cref{LEMMA:RECURRENCE-DEXP} by induction over $k$.  
  First consider the base case for $k=0$. In this case formula \eqref{EQ:DEXP-RECURSIVE} gives
  \begin{equation*}
    h^{(1)}=g^{(1)}h, 
  \end{equation*}
  which is easily seen to be true.

  Assuming formula \eqref{EQ:DEXP-RECURSIVE} holds for $h^{(k)}$ we  	
	continue with the case $k+1$.  
	Application of the product rule leads to
  \begin{align*}
    h^{(k+1)}=&\left(h^{(k)}\right)^{(1)}=\Bigg(\,\sum_{j=1}^{k}\binom{k-1}{j-1}g^{(j)}h^{(k-j)}\Bigg)^{(1)}\\
    =&\sum_{j=1}^{k}\binom{k-1}{j-1}g^{(j+1)}h^{(k-j)}+\sum_{j=1}^{k}\binom{k-1}{j-1}g^{(j)}h^{(k-j+1)}. 
  \end{align*}
  Substitution of $j'=j+1$ in the first term gives
  \begin{align*}
    h^{(k+1)}=&\sum_{j'=2}^{k+1}\binom{k-1}{j'-2}g^{(j')}h^{(k-j'+1)}+\sum_{j=1}^{k}\binom{k-1}{j-1}g^{(j)}h^{(k-j+1)}.
  \end{align*} 
  With this representation and the  identity
  \begin{equation*}
    \binom{k-1}{j-2}+\binom{k-1}{j-1}=\binom{k}{j-1}
  \end{equation*}
  we finally have
  \begin{align*}
    h^{(k+1)}=&g^{(1)}h^{(k)}+\sum_{j=2}^{k}\left[\binom{k-1}{j-1}+\binom{k-1}{j-2}\right]g^{(j)}h^{(k-j+1)}+g^{(k+1)}h\\
    =&\binom{k}{0}g^{(1)}h^{(k)}+\sum_{j=2}^{k}\binom{k}{j-1}g^{(j)}h^{(k-j+1)}+\binom{k}{k}g^{(k+1)}h\\
    =&\sum_{j=1}^{k+1}\binom{k}{j-1}g^{(j)}h^{(k-j+1)}. 
  \end{align*}
  This completes the proof of \Cref{LEMMA:RECURRENCE-DEXP}.
\end{proof}

\begin{proof}[Proof of \Cref{LEMMA:RECURSIVE-FORMULA-DELTAK}]  
To prove the recurrence formula \eqref{EQ:DELTA-RECURRENCE} we consider the characteristic function 
\begin{equation*}
        \varphi_{\rvIdnsII}(t)=\prod_{i=1}^{r}\left({1+\ccaC[i][2]t^{2}}\right)^{-\frac{1}{2}},\qquad
        t\in\R 
\end{equation*}
of the random variable \rvIdnsII\ introduced in \Cref{LEMMA:CF-INFODENSITY}. 
On the one hand, the series representation of $\varphi_{\rvIdnsII}$ given in \eqref{eq:sumvii} in the proof of \Cref{thm:pdfinf} can be rewritten using the coefficient $\delta_k$ defined in \eqref{EQ:DEF-DELTAK} in the following form. 
\begin{align}\label{EQ:ALT-SERIES-CF}
  \varphi_{\rvIdnsII}(t)=\left(1+\ccaC[r]^{2}t^{2}\right)^{-\frac{r}{2}}\left(\prod_{i=1}^{r-1}
    \frac{\ccaC[r]}{\ccaC[i]}\right)\sum_{\ell=0}^{\infty}\delta_{\ell}
  \left(1+\ccaC[r]^{2}t^{2}\right)^{-\ell},\quad t\in\R.
\end{align}
On the other hand, recall the expansion
given in \eqref{EQ:EXPAND-FACTORS-OF-CF-I}, which yields the identity
   \begin{align}\label{EQ:CHARACTERISTIC-FUNCTION-PROD-EXPENSION}
     \prod_{i=1}^{r}\left(1+\ccaC[i]^{2}t^{2}\right)^{-\frac{1}{2}}
     &=\left(1+\ccaC[r]^{2}t^{2}\right)^{-\frac{r}{2}}
       \left(\prod_{i=1}^{r-1}
       \frac{\ccaC[r]}{\ccaC[i]}\right) \prod_{i=1}^{r-1}
       \left(1+\left(\frac{\ccaC[r]^{2}}{\ccaC[i]^{2}}-1\right)
       \left(1+\ccaC[r]^{2}t^{2}\right)^{-1}\right)^{-\frac{1}{2}}.
   \end{align}
Applying the natural logarithm to both sides of \eqref{EQ:CHARACTERISTIC-FUNCTION-PROD-EXPENSION} gives
   \begin{equation}\label{EQ:LOG-CHARACTERISTIC-FUNCTION}
     \log\left(\varphi_{\rvIdnsII}(t)\right)
     =\log\left(\left(1+\ccaC[r]^{2}t^{2}\right)^{-\frac{r}{2}}
      \left(\prod_{i=1}^{r-1}
      \frac{\ccaC[r]}{\ccaC[i]}\right)\right)+\sum_{i=1}^{r-1}\log\left(
      \left(1+\left(\frac{\ccaC[r]^{2}}{\ccaC[i]^{2}}-1\right)
        \left(1+\ccaC[r]^{2}t^{2}\right)^{-1}\right)^{-\frac{1}{2}}
      \right).
   \end{equation} 
   Now consider the power series 
	\begin{equation}\label{EQ:POWER-SERIES-OF-LOG}
     \log(1+y)=\sum_{\ell=1}^{\infty}\frac{(-1)^{\ell+1}}{\ell}y^{\ell},
   \end{equation}
	which is absolutely convergent for $|y|<1$. 
With the same arguments as in the proof of \Cref{thm:pdfinf}, in particular due to \eqref{EQ:CONDITION-FOR-ABSOLUT-CONVERGENCE}, we can apply the series expansion \eqref{EQ:POWER-SERIES-OF-LOG} to the second term on the right-hand side of \eqref{EQ:LOG-CHARACTERISTIC-FUNCTION} to obtain the absolutely convergent series representation
   \begin{align}\label{EQ:SERIES-OF-LOG-CF}
	\begin{aligned}
     \log\left(\varphi_{\rvIdnsII}(t)\right)
     =&\log\left(\left(1+\ccaC[r]^{2}t^{2}\right)^{-\frac{r}{2}}
       \left(\prod_{i=1}^{r-1}
       \frac{\ccaC[r]}{\ccaC[i]}\right)\right)+\sum_{\ell=1}^{\infty}\sum_{i=1}^{r-1}\frac{1}{2\ell}
     \left(1-\frac{\ccaC[r]^{2}}{\ccaC[i]^{2}}\right)^{\ell}
     \left(1+\ccaC[r]^{2}t^{2}\right)^{-\ell}\\ 
     =&\log\left(\left(1+\ccaC[r]^{2}t^{2}\right)^{-\frac{r}{2}}
       \left(\prod_{i=1}^{r-1}
       \frac{\ccaC[r]}{\ccaC[i]}\right)\right)+\sum_{\ell=1}^{\infty}\gamma_{\ell}
     \left(1+\ccaC[r]^{2}t^{2}\right)^{-\ell}, 
		\end{aligned}
   \end{align}
 where we have used the definition of $\gamma_{\ell}$ given in \eqref{EQ:GAMMAK-DEF}. 
Applying the exponential function to both sides of \eqref{EQ:SERIES-OF-LOG-CF} then yields the  following expression for the characteristic function $ \varphi_{\rvIdnsII}$. 
   \begin{align}\label{EQ:SERIES-LOG-EXP-CF}
     \varphi_{\rvIdnsII}(t)
     =&\left(1+\ccaC[r]^{2}t^{2}\right)^{-\frac{r}{2}}\left(\prod_{i=1}^{r-1}\frac{\ccaC[r]}{\ccaC[i]}\right)
        \exp\left(\sum_{\ell=1}^{\infty}\gamma_{\ell}\left(1+\ccaC[r]^{2}t^{2}\right)^{-\ell}\right)
   \end{align} 
Comparing \eqref{EQ:ALT-SERIES-CF} and \eqref{EQ:SERIES-LOG-EXP-CF} yields the identity
   \begin{equation}\label{EQ:EXPCH-TERM}
    \sum_{\ell=0}^{\infty}\delta_{\ell}\left(1+\ccaC[r]^{2}t^{2}\right)^{-\ell} = \exp\left(\sum_{\ell=1}^{\infty}\gamma_{\ell}\left(1+\ccaC[r]^{2}t^{2}\right)^{-\ell}\right).
   \end{equation}
 We now define $x=\left(1+\ccaC[r]^{2}t^{2}\right)^{-1}$ and take the $(k+1)$-th derivative \wrt\ $x$ on both sides of \eqref{EQ:EXPCH-TERM} using the identity
  \begin{equation}\label{EQ:DERIVATIVE-OF-POWER-SERIES}
     \frac{\mathrm{d}^{m}}{\mathrm{d}x^{m}}\left(\sum_{\ell=0}^{\infty}a_{\ell}x^{\ell}\right)=\frac{\mathrm{d}^{m}}{\mathrm{d}x^{m}}\left(\sum_{\ell=1}^{\infty}a_{\ell}x^{\ell}\right)=
		\sum_{\ell=m}^{\infty}\frac{\ell!}{(\ell-m)!}a_{\ell}x^{\ell-m}
\end{equation}
for the $m$-th derivative of a power series $\sum_{\ell=0}^{\infty}a_{\ell}x^{\ell}$. 		
For the left-hand side of \eqref{EQ:EXPCH-TERM} we obtain
\begin{align}\label{EQ:DERIVATIVE-OF-POWER-SERIES-DELTAK}
\frac{\mathrm{d}^{k+1}}{\mathrm{d}x^{k+1}}\left(\sum_{\ell=0}^{\infty}\delta_{\ell}x^{\ell}\right)=\sum_{\ell=k+1}^{\infty}\frac{\ell!}{(\ell-k-1)!}\delta_{\ell}x^{\ell-k-1}.
\end{align}
For the right-hand side of \eqref{EQ:EXPCH-TERM} we obtain
\begin{align}\nonumber
\frac{\mathrm{d}^{k+1}}{\mathrm{d}x^{k+1}}&\left(\exp\left(\sum_{\ell=1}^{\infty}\gamma_{\ell}x^{\ell}\right)\right) \\\nonumber
&= \sum_{j=1}^{k+1}\binom{k}{j-1} \frac{\mathrm{d}^{j}}{\mathrm{d}x^{j}}\left(\sum_{\ell=1}^{\infty}\gamma_{\ell}x^{\ell}\right) \frac{\mathrm{d}^{k-j+1}}{\mathrm{d}x^{k-j+1}} \left(\exp\left(\sum_{\ell=1}^{\infty}\gamma_{\ell}x^{\ell}\right)\right)\\\nonumber
&= \sum_{j=1}^{k+1}\binom{k}{j-1} \frac{\mathrm{d}^{j}}{\mathrm{d}x^{j}}\left(\sum_{\ell=1}^{\infty}\gamma_{\ell}x^{\ell}\right) \frac{\mathrm{d}^{k-j+1}}{\mathrm{d}x^{k-j+1}} \left(  \sum_{\ell=0}^{\infty}\delta_{\ell}x^{\ell} \right)\\\label{EQ:DERIVATIVE-OF-POWER-SERIES-GAMMAK}
&= \sum_{j=1}^{k+1}\binom{k}{j-1} \left(\sum_{\ell=j}^{\infty}\frac{\ell!}{(\ell-j)!}\gamma_{\ell}x^{\ell-j}\right) \left(  \sum_{\ell=k+1-j}^{\infty} \frac{\ell!}{(\ell-k+j-1)!} \delta_{\ell}x^{\ell-k+j-1} \right),
\end{align} 
where we used \Cref{LEMMA:RECURRENCE-DEXP} for the first equality and the identities \eqref{EQ:EXPCH-TERM} and \eqref{EQ:DERIVATIVE-OF-POWER-SERIES} for the second and third. 
From the equality
\begin{align*}
\frac{\mathrm{d}^{k+1}}{\mathrm{d}x^{k+1}}\left(\sum_{\ell=0}^{\infty}\delta_{\ell}x^{\ell}\right) = \frac{\mathrm{d}^{k+1}}{\mathrm{d}x^{k+1}}\left(\exp\left(\sum_{\ell=1}^{\infty}\gamma_{\ell}x^{\ell}\right)\right)
\end{align*}
and the evaluation of the right-hand side of \eqref{EQ:DERIVATIVE-OF-POWER-SERIES-DELTAK} and \eqref{EQ:DERIVATIVE-OF-POWER-SERIES-GAMMAK} we obtain
\begin{alignat*}{3}
(k+1)!\delta_{k+1} &\,x^0+\big(\ldots\big)&&\,x^1+\big(\ldots\big)&&\,x^2\ldots\,\\
=\left(\sum_{j=1}^{k+1}\binom{k}{j-1}j!\gamma_j(k+1-j)!\delta_{k+1-j}\right)&\,x^0+\big(\ldots\big)&&\,x^1+\big(\ldots\big)&&\,x^2\ldots
\end{alignat*}
Comparing the coefficients for $x^0$ yields 
\begin{align*}
\delta_{k+1}&=\frac{1}{(k+1)!}\sum_{j=1}^{k+1}\binom{k}{j-1}\,j!\,\gamma_j(k+1-j)!\,\delta_{k+1-j}\\
&=\frac{1}{(k+1)!}\sum_{j=1}^{k+1} \frac{k!}{(j-1)!(k+1-j)!}\, j!\,\gamma_j(k+1-j)!\,\delta_{k+1-j}\\
&=\frac{1}{(k+1)}\sum_{j=1}^{k+1} j\,\gamma_j\,\delta_{k+1-j},
\end{align*}
which proves the recursive formula of \Cref{LEMMA:RECURSIVE-FORMULA-DELTAK} to be shown.  %
\end{proof}

\begin{remark}[Relation to Bell polynomials]\
Interestingly, the coefficient $\delta_k$ can be expressed for all $k\in\N$ in the following form 
\begin{align*}
\delta_{k} = \frac{B_k\big(\gamma_1,2\gamma_2,6\gamma_3,\ldots,k!\gamma_k\big)}{k!}, 
\end{align*}
where $\gamma_j$ is defined in \eqref{EQ:GAMMAK-DEF} and $B_k$ denotes the complete Bell polynomial of order $k$ 
\cite[Sec.\,3.3]{Comtet1974}. 
Even though this is an interesting connection to the Bell polynomials, which provides an explicit formula of $\delta_k$, the recursive formula given in \Cref{LEMMA:RECURSIVE-FORMULA-DELTAK} is more efficient for numerical calculations. 
\end{remark}

\begin{remark}[Finite sum approximation of alternative representation of PDF and CDF]\label{RMK:ALT-FINITE-SUM-APPROXIMATION} The results of \Cref{LEMMA:RECURRENCE-UN}, \ref{LEMMA:RECURRENCE-DN} and \ref{LEMMA:RECURSIVE-FORMULA-DELTAK} can be used in the following way for numerical calculations.
Consider
\begin{align}\label{EQ:APPROX-PDF-RECURSIVE}
\hat{f}_{\iDn(\rvX;\rvY)}(x,n)=\frac{1}{\ccaC[r]\sqrt{\pi}}\left(\prod_{i=1}^{r-1}\frac{\ccaC[r]}{\ccaC[i]}\right)\sum_{k=0}^{n}\delta_{k}\mathrm{U}_{k}\left(\left|\frac{x-I(\xi;\eta)}{\ccaC[r]}\right|\right),\qquad x\in\R
\end{align}
 for $n\in\No$, i.\,e., the finite sum approximation of the PDF given in \eqref{EQ:PDF-RECURSIVE}.  
To calulate $\hat{f}_{\iDn(\rvX;\rvY)}(x,n)$ first calculate $\mathrm{U}_{0}\big(|x-I(\xi;\eta)|/\ccaC[r]\big)$ and $\mathrm{U}_{1}\big(|x-I(\xi;\eta)|/\ccaC[r]\big)$
using \eqref{EQ:DEFINITION-FUNCTION-UN}. Then use the recurrence formulas  \eqref{EQ:RECURRENCE-FUNKTION-UN} and \eqref{EQ:DELTA-RECURRENCE} to calculate the remaining summands  in \eqref{EQ:APPROX-PDF-RECURSIVE}.  
The great advantage of this approach is that only two  evaluations of the modified Bessel function are required and for the rest of the calculations efficient recursive formulas are employed making the numerical computations very efficient. 

Similarly, consider 
    \begin{equation*}
      \hat{F}_{\iDn(\xi;\eta)}(x,n)=
      \begin{dcases}
        \rule{0ex}{3.5ex}\;\frac{1}{2}-\hat{\funcV}\left(I(\xi;\eta)-x,n\right)&\text{if}\quad x \leq I(\xi;\eta)\\
        \;\frac{1}{2}+\hat{\funcV}\left(x-I(\xi;\eta),n\right)&\text{if}\quad x > I(\xi;\eta)\\[1ex]
      \end{dcases},
    \end{equation*}
with 
 \begin{align}\label{EQ:CDF-RECURSIVE-FUNCTION-V-APPROX}
  \hat{\funcV}(z,n)=\left(\prod_{i=1}^{r-1}\frac{\ccaC[r]}{\ccaC[i]}\right)\sum_{k=0}^{n}\delta_{k}\mathrm{D}_{k}(z),\qquad z\geq 0,
 \end{align}
for $n\in\No$, i.\,e., the finite sum approximation  of the alternative representation of the CDF of the information density, where $\hat{\funcV}(z,n)$ is the finite sum approximation of the function $\funcV(\cdot)$ given in \eqref{EQ:CDF-RECURSIVE-FUNCTION-V}.  
To calculate $\hat{F}_{\iDn(\rvX;\rvY)}(x,n)$ first calculate $\mathrm{D}_{0}(z)$, $\mathrm{U}_{0}\left(\frac{z}{\ccaC[r]}\right)$, and $\mathrm{U}_{1}\left(\frac{z}{\ccaC[r]}\right)$ for $z=\mInf{\rvX}{\rvY}-x$ or $z=x-\mInf{\rvX}{\rvY}$ using \eqref{EQ:DEFINITION-FUNCTION-UN} and \eqref{EQ:DEFINITION-FUNCTION-DN}. 
Then use the recurrence formulas  \eqref{EQ:RECURRENCE-FUNKTION-UN}, \eqref{EQ:RECURRENCE-FUNKTION-DN} and \eqref{EQ:DELTA-RECURRENCE} to calculate the remaining summands  in \eqref{EQ:CDF-RECURSIVE-FUNCTION-V-APPROX}. This approach requires only three evaluations of the modified Bessel and Struve $\mathrm{L}$ function resulting in very efficient numerical calculations also for the CDF of the information density. 
\end{remark}

To control the approximation error of $\hat{f}_{\iDn(\rvX;\rvY)}(x,n)$ and $\hat{F}_{\iDn(\rvX;\rvY)}(x,n)$ error bounds similar to those in \Cref{THEOREM:APPROXIMATION-ERROR-PDF-CDF} can be utilized.  The subsequent \namecref{PROP:APPROXIMATION-ERROR-RECURSIVE} provides such bounds. 

\newpage

\begin{proposition}[Bounds of the approximation error for the alternative representation of PDF and CDF]\label{PROP:APPROXIMATION-ERROR-RECURSIVE} 
	For the finite sum approximations in \Cref{RMK:ALT-FINITE-SUM-APPROXIMATION}
  of the  alternative representation of the PDF and CDF of the information density as given in \Cref{PROP:ALT-SERIES-REPRESENTATIONS-OF-PDF-CDF} we have for $n\in\N$ summands the  error bounds
  \begin{align*}
    \big|{f}_{\iDn(\rvX;\rvY)}(x)-\hat{f}_{\iDn(\rvX;\rvY)}(x,n)\big|\leq
    \frac{\Gamma\left(\frac{r-1}{2}+n\right)}{2\ccaC[r]\sqrt{\pi}\,\Gamma\left(\frac{r}{2}+n\right)}\left(1-\sum_{k=0}^{n}\delta_{k}\right),
\qquad x\in\R
\end{align*}
and 
\begin{align*}
  \big|\funcV(z)-\hat{\funcV}(z,n)\big|\leq
  \frac{1}{2}\left(1-\sum_{k=0}^{n}\delta_{k}\right), \qquad z\geq 0.
\end{align*}
\end{proposition}

\begin{proof} The proof of \Cref{PROP:APPROXIMATION-ERROR-RECURSIVE} follows along the same lines as the proof of \Cref{THEOREM:APPROXIMATION-ERROR-PDF-CDF} using the identity
\begin{align*}
     \left(\prod_{i=1}^{r-1}
      \frac{\ccaC[r]}{\ccaC[i]}\right)\sum_{k=0}^{\infty}\delta_k
	= 
	\sum_{k_{2}=0}^{\infty}\dots
    \sum_{k_{r-1}=0}^{\infty}\left[\prod_{i=1}^{r-1}
      \frac{\ccaC[r]}{\ccaC[i]}\frac{(2k_{i})!}{(k_{i}!)^{2}4^{k_{i}}}
      \left(1-\frac{\ccaC[r][2]}{\ccaC[i][2]}\right)^{k_{i}}\right] =1, 
\end{align*}
which follows from \eqref{EQ:SERIES-OF-PRODUCTS} and the definition of the coefficient $\delta_k$ in \eqref{EQ:DEF-DELTAK}.  
\end{proof}

\section{Numerical Examples and Illustrations}
\label{SEC:EXAMPLES-AND-ILLUSTRATIONS}
We illustrate the results of this paper with some examples, which all can be verified with the Python implementation publicly available on GITLAB \cite{HuffmannGitlab2021}.  
First, we consider the special case of \Cref{COR:PDF-CDF-EQUAL-CORRELATIONS} when all canonical correlations are equal.  
The PDF and CDF given by \eqref{EQ:PDF-INFO-DENSITY-EQUAL-CCA} and \eqref{EQ:CDF-INFO-DENSITY-EQUAL-CCA} are illustrated in \Cref{FIGURE:ILLUSTRATION-PDF-CASE-I} and 
\ref{FIGURE:ILLUSTRATION-CDF-CASE-I} in centered from, \ie, shifted by $I(\rvX;\rvY)$,  for $r\in\{1,2,3,4,5\}$ and equal canonical correlations $\ccaC[i]=0.9, i=1,\ldots,r$.  
In \Cref{FIGURE:ILLUSTRATION-PDF-CASE-II} and \ref{FIGURE:ILLUSTRATION-CDF-CASE-II} 
 a fixed number of $r=5$ equal canonical correlations $\ccaC[i]\in\{0.1,0.2,0.5,0.7,0.9\}, i=1,\ldots,r$ is considered.    
When all canonical correlations are equal, then  due to  the central limit theorem the distribution of the information density $\iDn(\rvX;\rvY)$ converges to a Gaussian distribution as $r\rightarrow\infty$.  
\Cref{FIGURE:ILLUSTRATION-PDF-CASE-III} and 
\ref{FIGURE:ILLUSTRATION-CDF-CASE-III} show for $r\in\{5,10,20,40\}$ and equal canonical correlations $\ccaC[i]=0.2,r=1,2,\ldots,r$ the PDF and CDF of the information density together with corresponding Gaussian approximations.  
The approximations are obtained by considering Gaussian distributions,  which have the same variance as the information density $\iDn(\rvX;\rvY)$. Recall, the variance of the information density is given by \eqref{EQ:VARIANCE-OF-INFO-DENSITY}, \ie, by the sum of the squared canonical correlations.   
The illustrations show that only for a high number of equal canonical correlations  the distribution of the information density becomes  approximately Gaussian.

To illustrate the case with different canonical correlations let us consider the sequence $\{\ccaC[1](T),\ccaC[2](T),$ $\ldots,\ccaC[r](T)\}$ with
\begin{align}\label{EQ:CCA-OU-AWGN}
\ccaC[i](T)=\sqrt{\frac{T^2}{T^2+\pi\left(i-\frac{1}{2}\right)^2}},\qquad i=1,2,\ldots,r.
\end{align}
These canonical correlations are related to the information density of a continuous-time additive white Gaussian noise channel confined to a finite time interval $[0,T]$ with a Brownian motion as input signal (see e.\,g.\ Huffmann \cite[Sec.\,8.1]{Huffmann2021} for more details). 
\Cref{FIGURE:ILLUSTRATION-PDF-CASE-IV} and \ref{FIGURE:ILLUSTRATION-CDF-CASE-IV} show the approximated PDF $\hat{f}_{\iDn(\rvX;\rvY)-I(\rvX;\rvY)}(\cdot,n)$ and CDF $\hat{F}_{\iDn(\rvX;\rvY)-I(\rvX;\rvY)}(\cdot,n)$  for $r\in\{2,5,10,15\}$ and $T=1$ using the finite sums   
\eqref{EQ:APPROX-PDF-RECURSIVE} and \eqref{EQ:CDF-RECURSIVE-FUNCTION-V-APPROX}. The bounds of the approximation error given in   
\Cref{PROP:APPROXIMATION-ERROR-RECURSIVE} 
are chosen $<1\mathrm{e}$-$2$ such that there are no differences visible in the plotted curves by further lowering the approximation error.   
The number $n$ of summands required in \eqref{EQ:APPROX-PDF-RECURSIVE} and \eqref{EQ:CDF-RECURSIVE-FUNCTION-V-APPROX} to achieve these error bounds for $r\in\{2,5,10,15\}$ is equal to $n\in\{15,141,638,1688\}$ for the PDF and $n\in\{20,196,886,2071\}$ for the CDF. 
Choosing $r$ larger than $15$ for the canonical correlations \eqref{EQ:CCA-OU-AWGN} with $T=1$ does not result in visible changes of the PDF and CDF compared to $r=15$.  
This demonstrates together with \Cref{FIGURE:ILLUSTRATION-PDF-CASE-IV} and \ref{FIGURE:ILLUSTRATION-CDF-CASE-IV} that a Gaussian approximation is not valid for this example, even if $r\rightarrow\infty$.   
%

Indeed, from \cite[Th.\,9.6.1]{Pinsker1964} and the comment above equation (9.6.45) in  \cite{Pinsker1964}  one can conclude that whenever the canonical correlations satisfy
\begin{align*}
\lim_{r\rightarrow\infty}\,\sum_{i=1}^r\ccaC[i][2] < \,\infty,
\end{align*}
then the distribution of the information density is \emph{not} Gaussian.

\section{Summary of Contributions}
\label{SECTION:CONLUSIONS}

In this paper we derived series representations of the PDF and CDF of the information density  for arbitrary Gaussian random vectors as well as a general formula for the central moments using  canonical correlation analysis.  We provided simplified and closed-form expressions for important special cases, in particular when all canonical correlations are equal, and derived error bounds for finite sum approximations of the general series representations.   
These approximations are suitable for arbitrarily accurate numerical calculations, where the approximation error can be easily controlled  with the derived error bounds.  
Furthermore, we derived recurrence formulas, which allow very efficient numerical calculations of the PDF and CDF. 
Moreover, we provided examples showing the (in)validity of approximating the information density with a Gaussian random variable.

\newpage

\begin{figure}[h!]
\centering
\includegraphics[width=0.85\columnwidth]{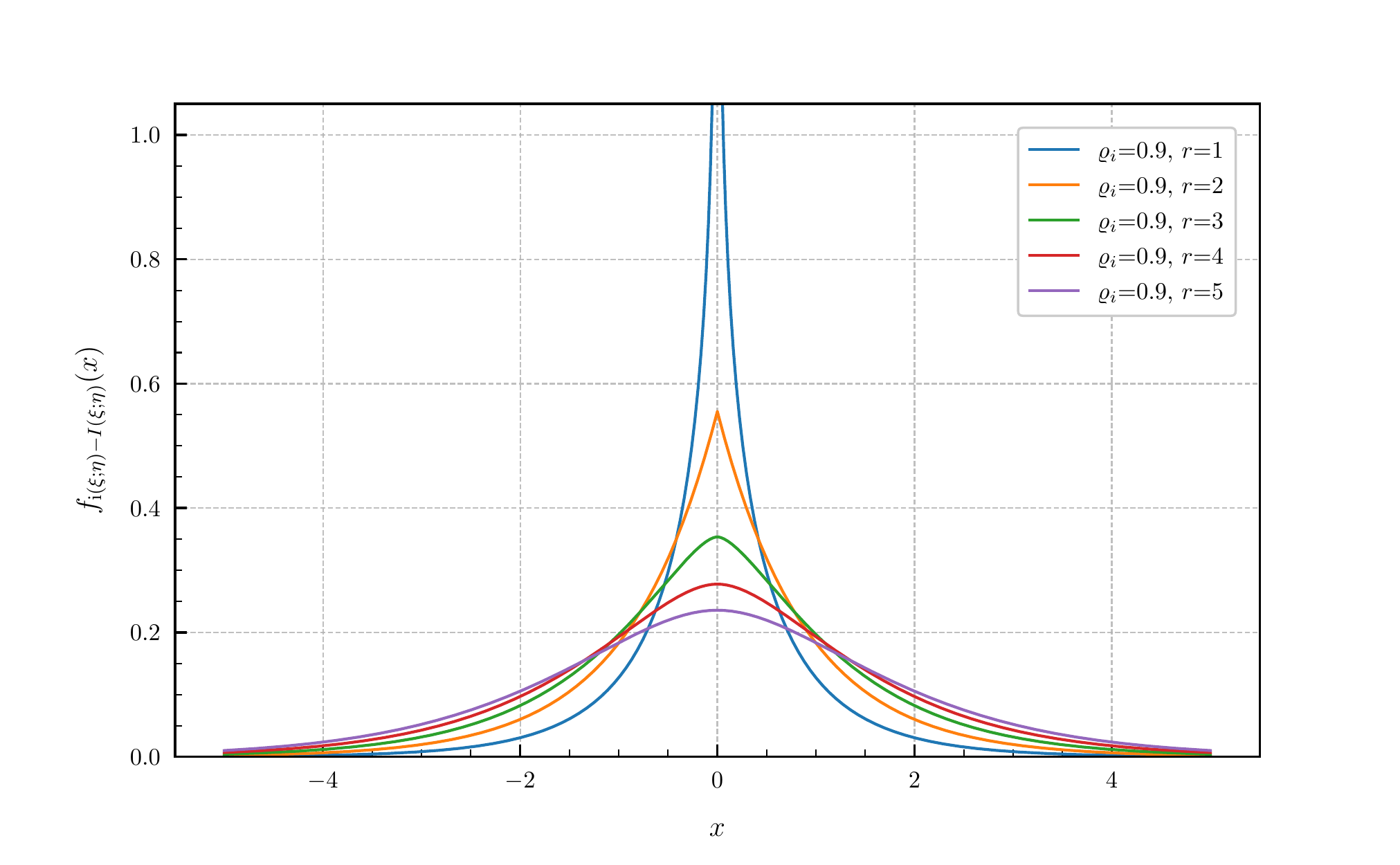}%
\caption{PDF $f_{\iDn(\rvX;\rvY)-I(\rvX;\rvY)}$ for $r\in\{1,2,3,4,5\}$ equal canonical correlations $\ccaC[i]=0.9$.}%
\label{FIGURE:ILLUSTRATION-PDF-CASE-I}%
\end{figure}

\begin{figure}[h!]
\centering
\includegraphics[width=0.85\columnwidth]{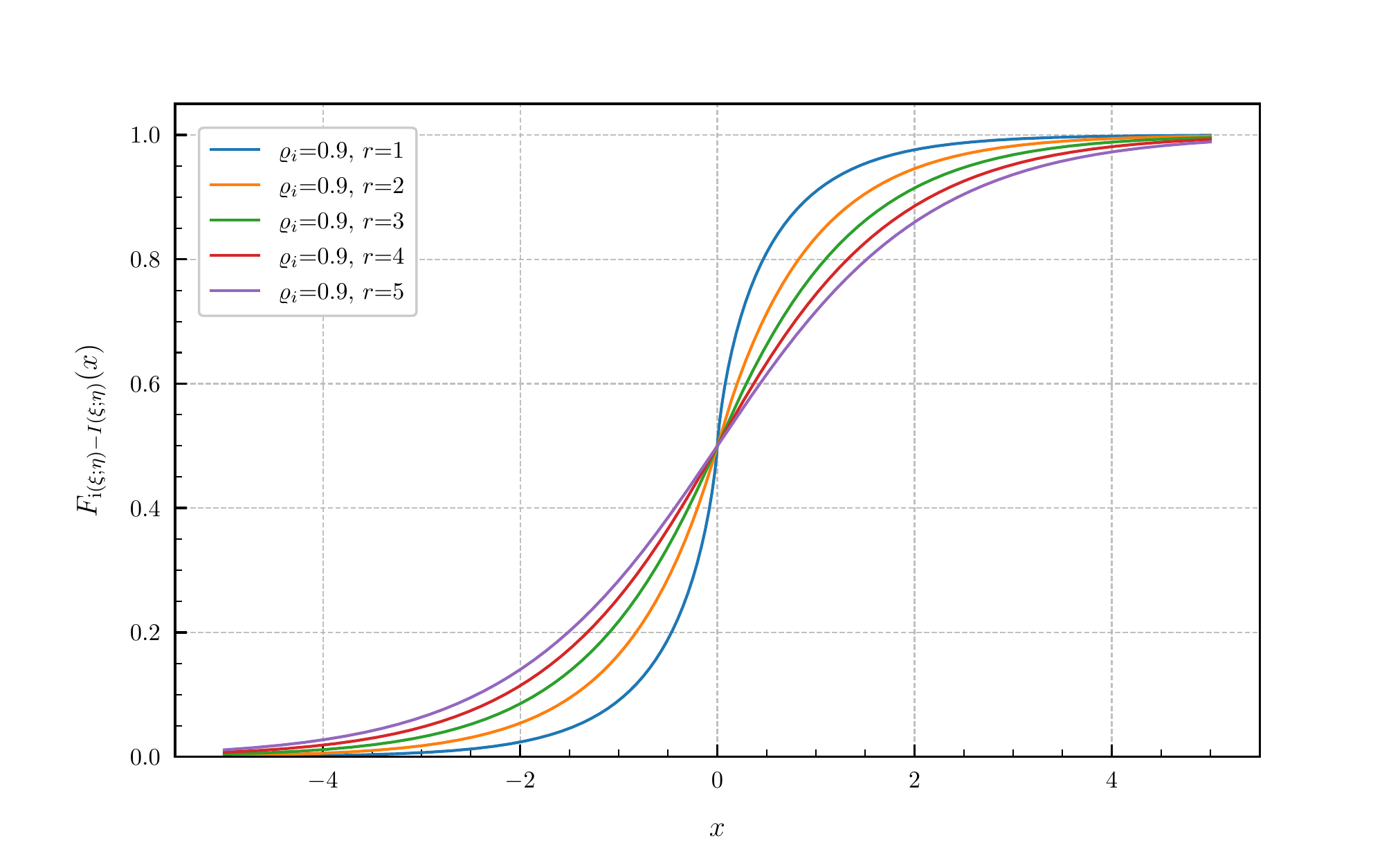}%
\caption{CDF $F_{\iDn(\rvX;\rvY)-I(\rvX;\rvY)}$ for $r\in\{1,2,3,4,5\}$ equal canonical correlations $\ccaC[i]=0.9$.}%
	\label{FIGURE:ILLUSTRATION-CDF-CASE-I}
\end{figure}

\newpage

\begin{figure}[h!]
\centering
\includegraphics[width=0.85\columnwidth]{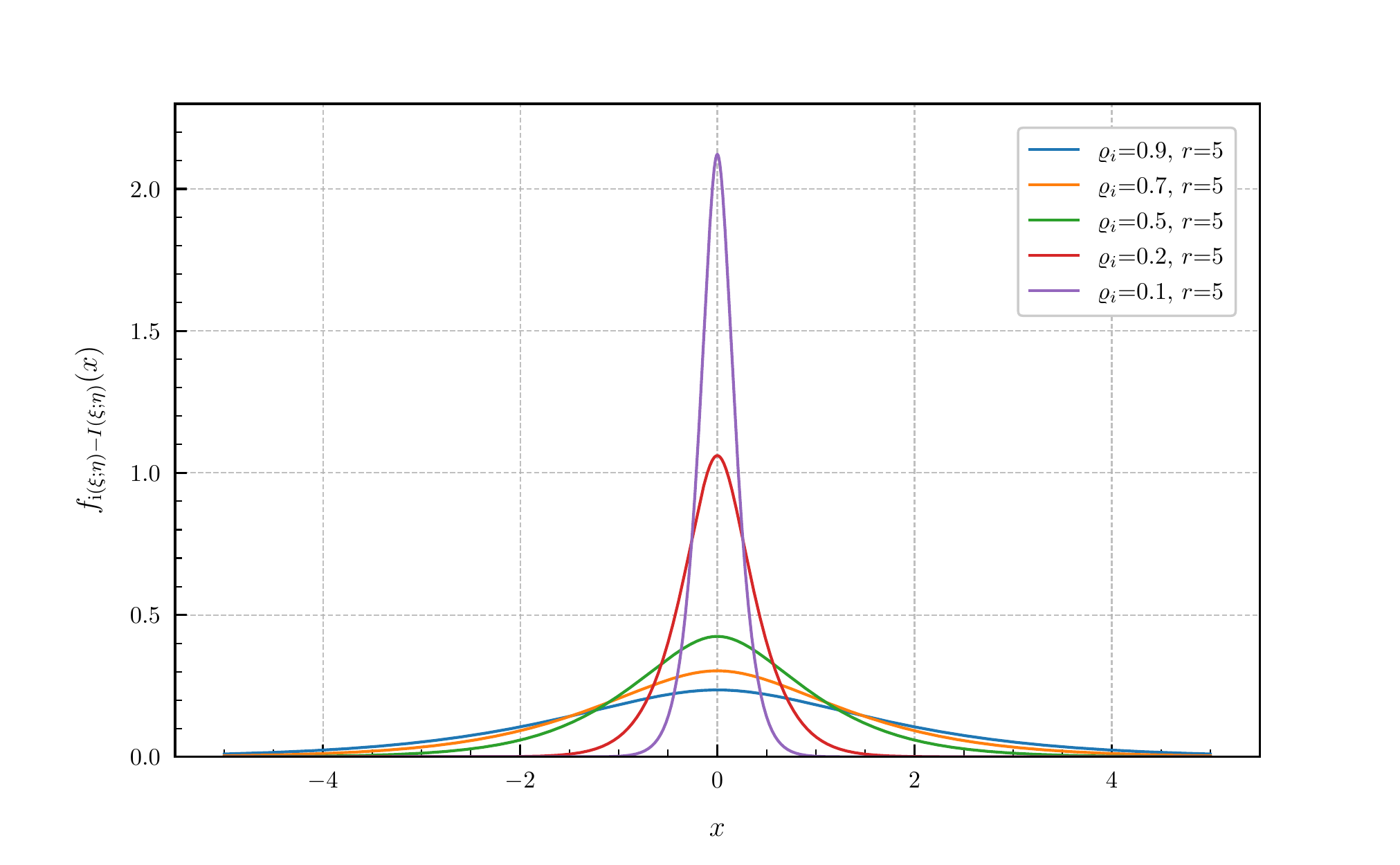}%
\caption{PDF  $f_{\iDn(\rvX;\rvY)-I(\rvX;\rvY)}$ for $r=5$ equal canonical correlations $\ccaC[i]\in\{0.1,0.2,0.5,0.7,0.9\}$.}%
\label{FIGURE:ILLUSTRATION-PDF-CASE-II}%
\end{figure}

\begin{figure}[h!]
\centering
\includegraphics[width=0.85\columnwidth]{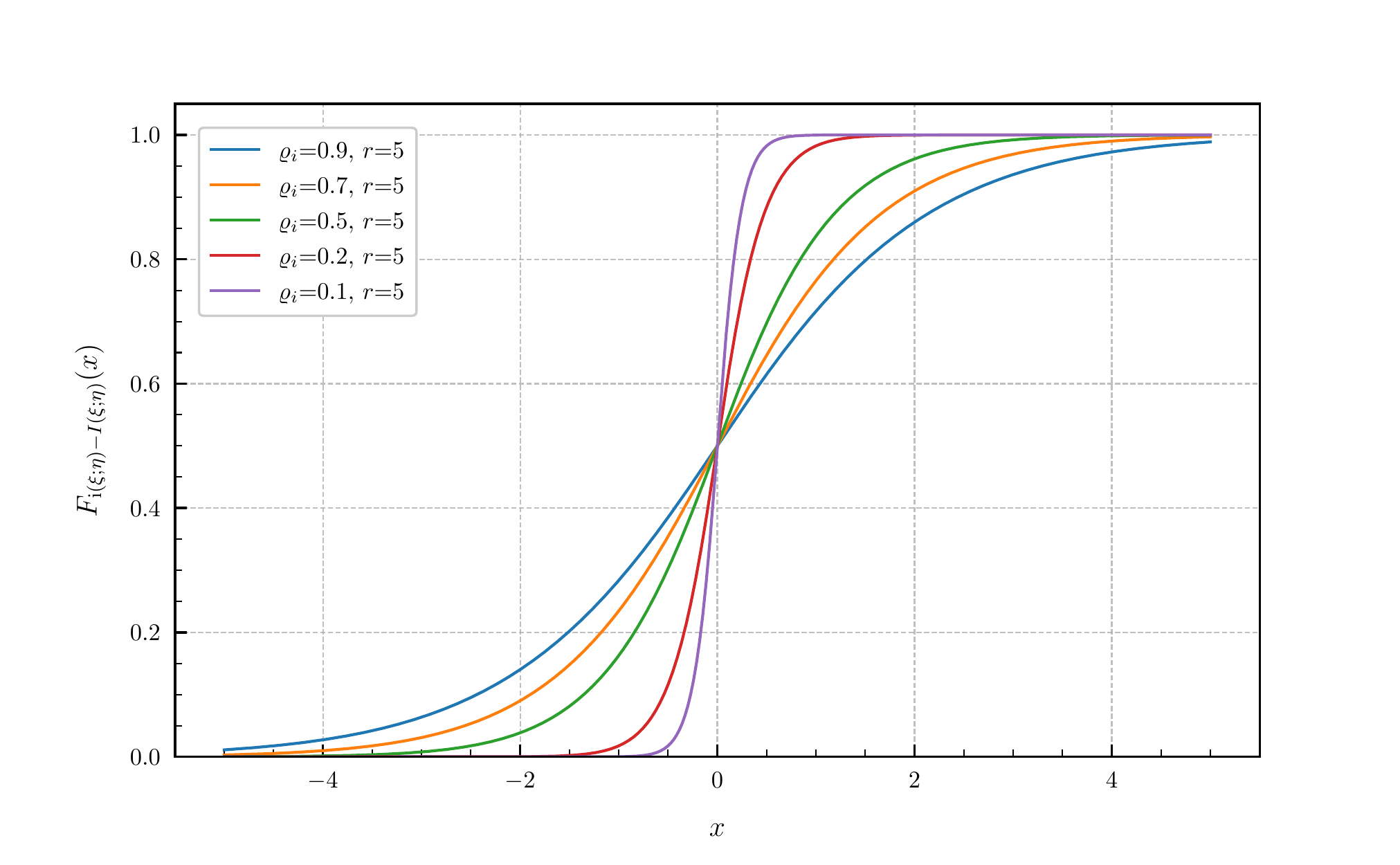}%
\caption{CDF  $F_{\iDn(\rvX;\rvY)-I(\rvX;\rvY)}$ for $r=5$ equal canonical correlations $\ccaC[i]\in\{0.1,0.2,0.5,0.7,0.9\}$.}%
	\label{FIGURE:ILLUSTRATION-CDF-CASE-II}
\end{figure}

\newpage

\begin{figure}[h!]
\centering
\includegraphics[width=0.85\columnwidth]{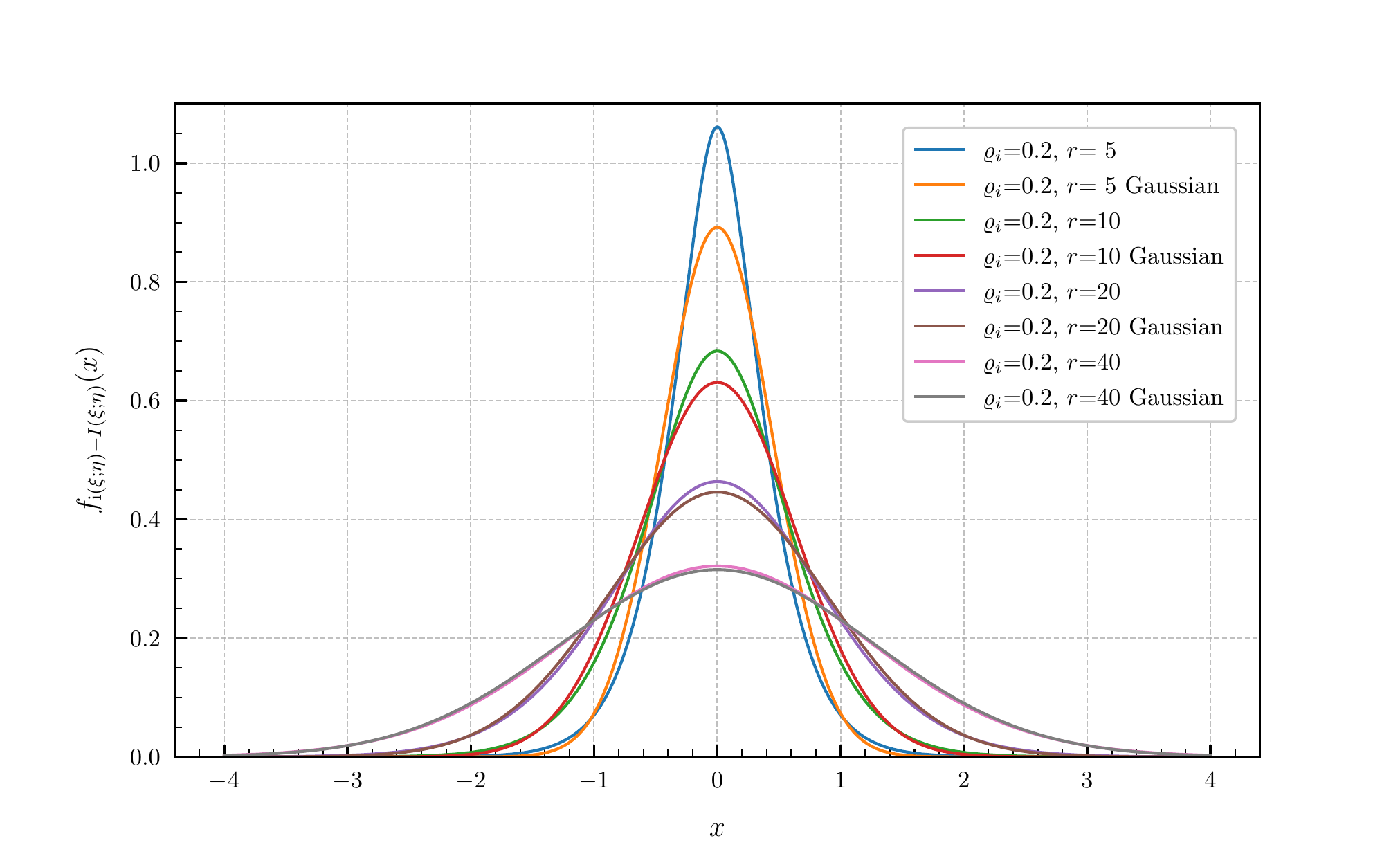}%
\caption{PDF $f_{\iDn(\rvX;\rvY)-I(\rvX;\rvY)}$ for $r\in\{5,10,20,40\}$ equal canonical correlations $\ccaC[i]=0.2$ vs. Gaussian approximation.}%
\label{FIGURE:ILLUSTRATION-PDF-CASE-III}%
\end{figure}
\begin{figure}[h!]
\centering
\includegraphics[width=0.85\columnwidth]{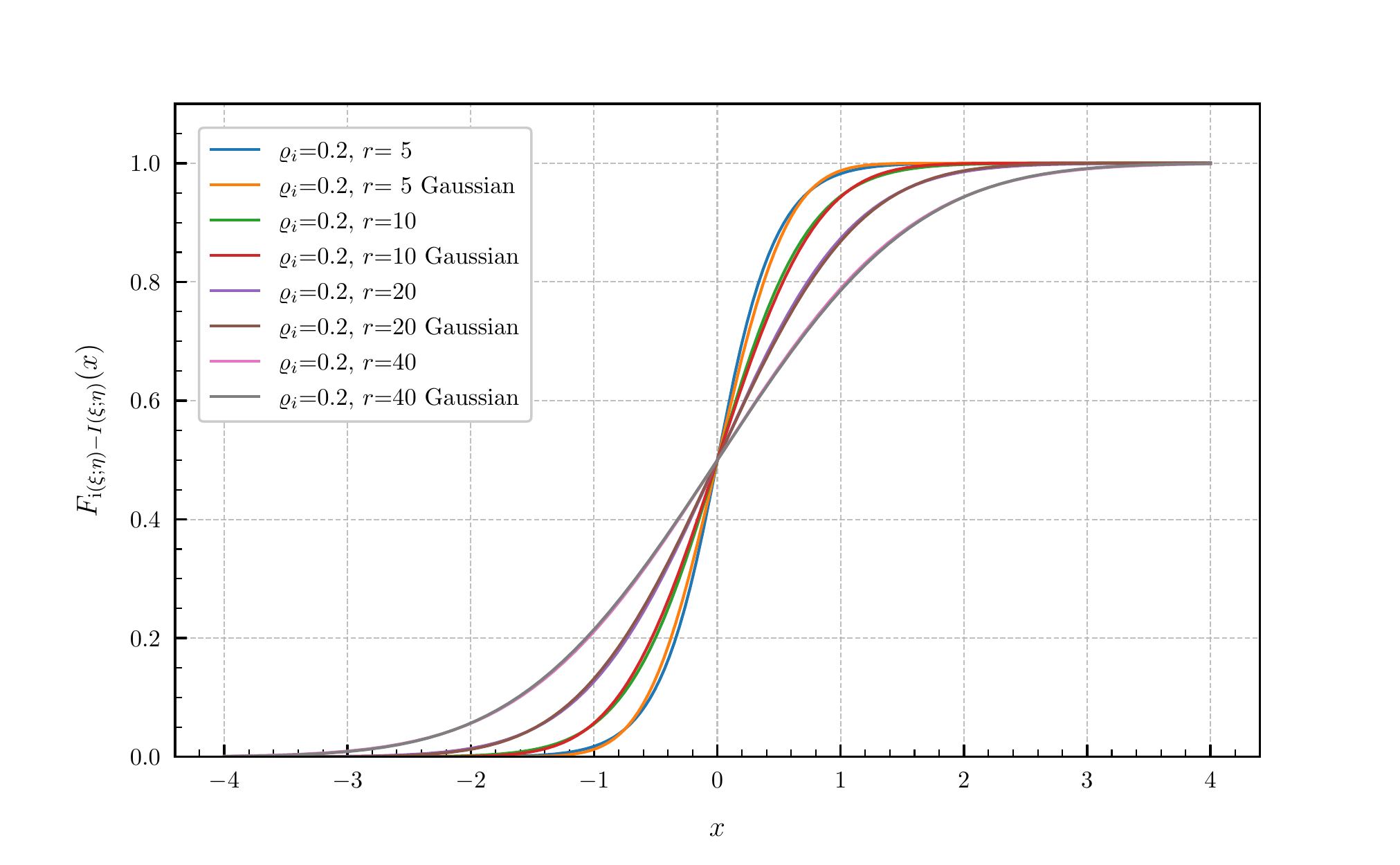}%
\caption{CDF $F_{\iDn(\rvX;\rvY)-I(\rvX;\rvY)}$ for $r\in\{5,10,20,40\}$ equal canonical correlations $\ccaC[i]=0.2$ vs. Gaussian approximation.}%
	\label{FIGURE:ILLUSTRATION-CDF-CASE-III}
\end{figure}

\newpage

\begin{figure}[h!]
\centering
\includegraphics[width=0.85\columnwidth]{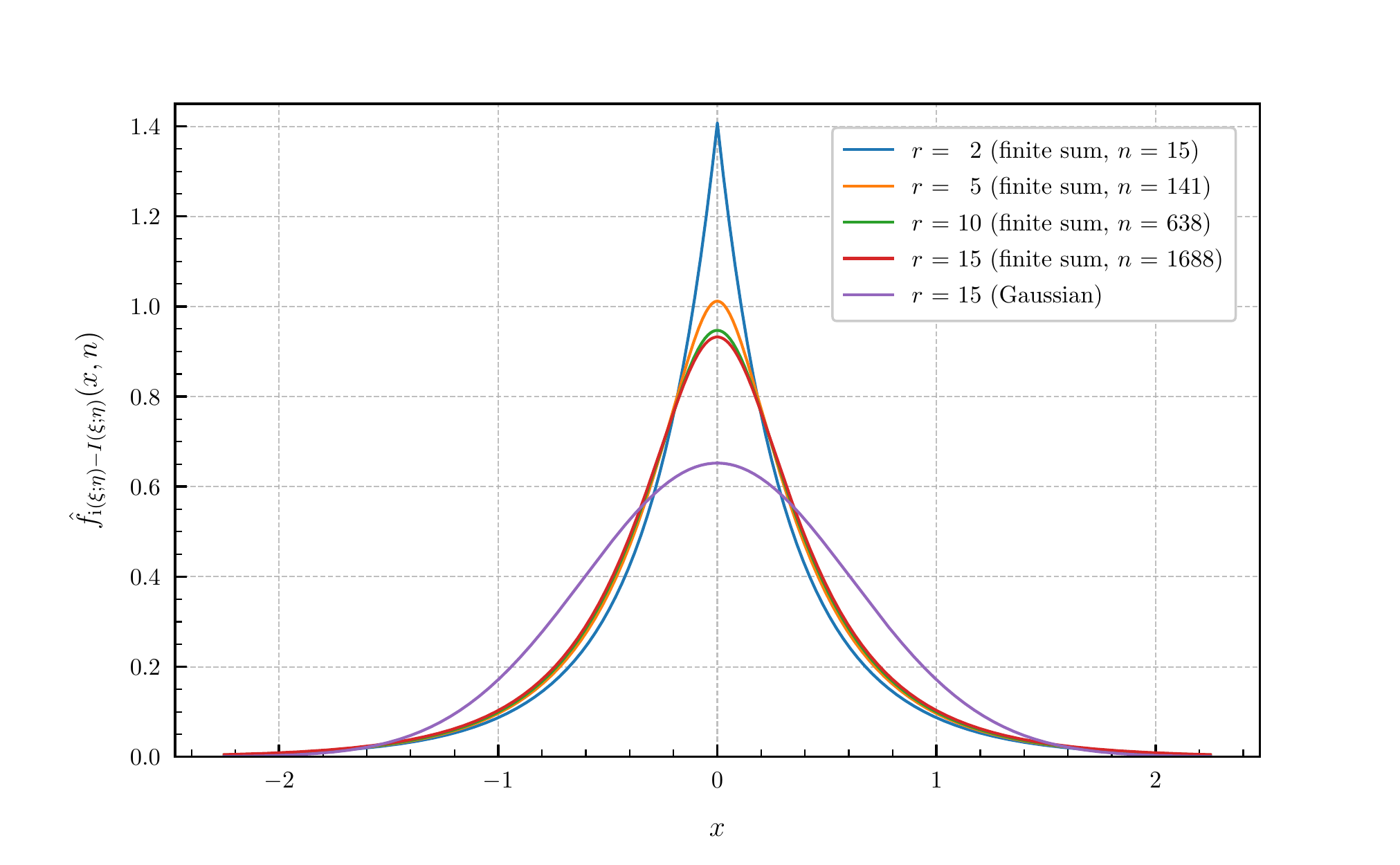}%
\caption{Approximated PDF $\hat{f}_{\iDn(\rvX;\rvY)-I(\rvX;\rvY)}(\cdot,n)$ for $r\in\{2,5,10,15\}$ canonical correlations $\ccaC[i](T)$ given in \eqref{EQ:CCA-OU-AWGN} for $T=1$ (approximation error $<1\mathrm{e}\text{-}02$) vs. Gaussian approximation ($r=15$).}%
\label{FIGURE:ILLUSTRATION-PDF-CASE-IV}%
\end{figure}
\begin{figure}[h!]
\centering
\includegraphics[width=0.85\columnwidth]{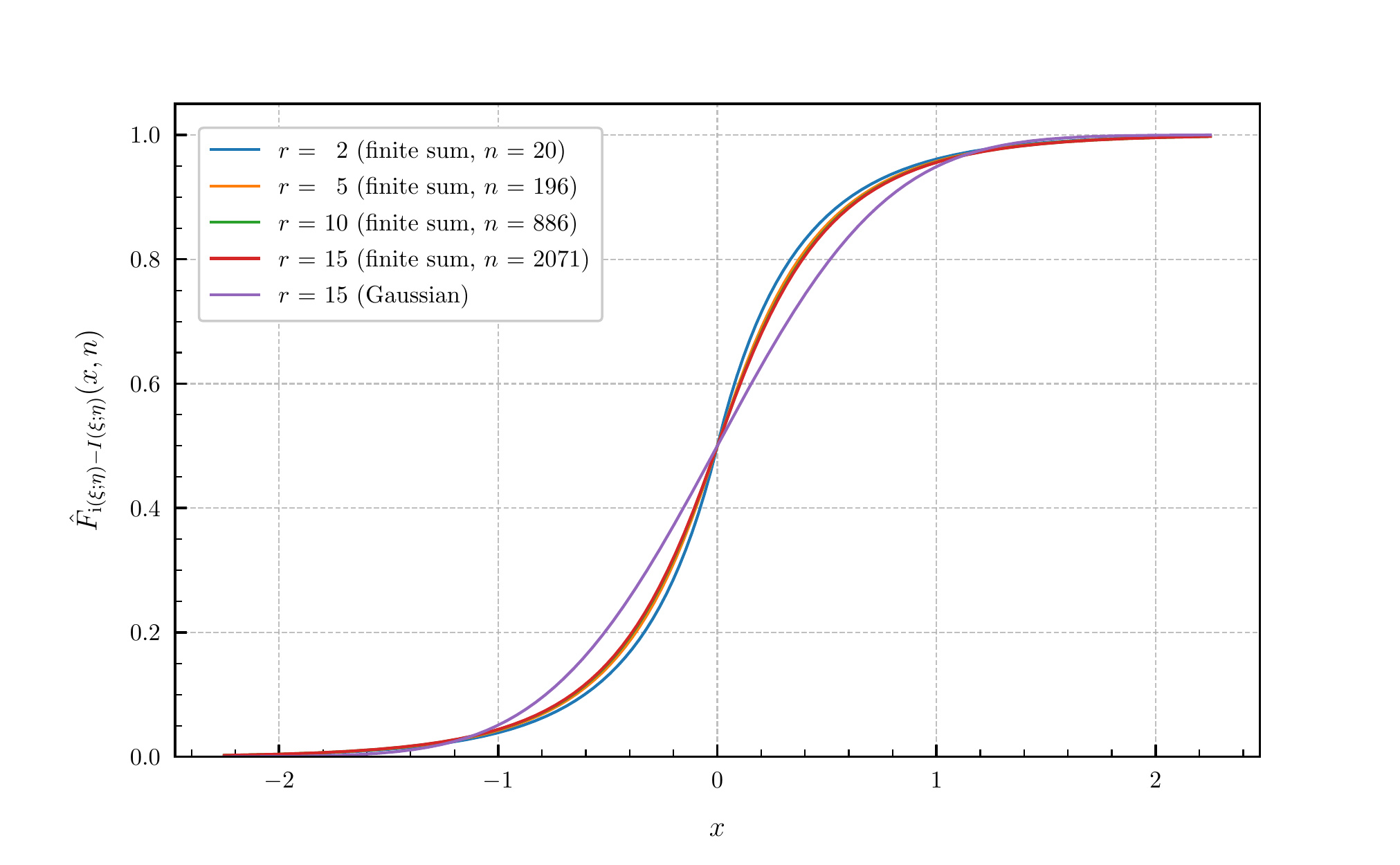}%
\caption{Approximated CDF $\hat{F}_{\iDn(\rvX;\rvY)-I(\rvX;\rvY)}(\cdot,n)$ for $r\in\{2,5,10,15\}$ canonical correlations $\ccaC[i](T)$ given in \eqref{EQ:CCA-OU-AWGN} for $T=1$ (approximation error $<1\mathrm{e}\text{-}02$) vs. Gaussian approximation ($r=15$).}%
	\label{FIGURE:ILLUSTRATION-CDF-CASE-IV}
\end{figure}


\newpage

\bibliographystyle{IEEEtran}
\bibliography{articles,books}

\IEEEpubidadjcol

\end{document}